\documentclass[%
 reprint,
superscriptaddress,
dvipdfmx,
 amsmath,amssymb,
 aps,
]{revtex4-2}

\usepackage[utf8]{inputenc}
\usepackage{graphicx}
\usepackage{color}
\usepackage{braket}
\usepackage{physics}
\usepackage{comment}
\usepackage{mathrsfs}
\usepackage{enumerate}
\usepackage{here}

\usepackage[version=4]{mhchem}
\usepackage[dvipdfmx,breaklinks=true]{hyperref}

\usepackage{amsthm}
\theoremstyle{definition}
\newtheorem{theorem}{Theorem}[]
\newtheorem{lemma}[theorem]{Lemma}

\usepackage{footmisc}

\usepackage{mathtools} %
\mathtoolsset{showonlyrefs,showmanualtags}

\usepackage[ruled,vlined]{algorithm2e}
\usepackage{amsthm}

\begin{document}
\title{Analytical construction of $(n, n-1)$ quantum random access codes saturating the conjectured bound}

\author{Takayuki Suzuki}
\affiliation{Technology Strategy Division, SCSK Corporation, Koto, Tokyo 135-8110, Japan}

\begin{abstract}
Quantum Random Access Codes (QRACs) embody the fundamental trade-off between the compressibility of information into limited quantum resources and the accessibility of that information, serving as a cornerstone of quantum communication and computation.
In particular, the $(n, n-1)$-QRACs, which encode $n$ bits of classical information into $n-1$ qubits, provides an ideal theoretical model for verifying quantum advantage in high-dimensional spaces; however, the analytical derivation of optimal codes for general $n$ has remained an open problem.
In this paper, we establish an analytical construction method for $(n, n-1)$-QRACs by using an explicit operator formalism.
We prove that this construction strictly achieves the numerically conjectured upper bound of the average success probability, $\mathcal{P} = 1/2 + \sqrt{(n-1)/n}/2$, for all $n$.
Furthermore, we present a systematic algorithm to decompose the derived optimal POVM into standard quantum gates. Since the resulting decoding circuit consists solely of interactions between adjacent qubits, it can be implemented with a circuit depth of $O(n)$ even under linear connectivity constraints.
Additionally, we analyze the high-dimensional limit and demonstrate that while the non-commutativity of measurements is suppressed, an information-theoretic gap of $O(\log n)$ from the Holevo bound inevitably arises for symmetric encoding.
This study not only provides a scalable implementation method for high-dimensional quantum information processing but also offers new insights into the mathematical structure at the quantum-classical boundary.
\end{abstract}

\maketitle

\section{Introduction}\label{Sec:Introduction}

Understanding the fundamental limits of information processing in quantum mechanical systems is a central theme in quantum information theory.
While the Holevo bound~\cite{holevo1973} demonstrated that a system of $n$ qubits cannot be used to retrieve more than $n$ bits of classical information, the selective readout of that information exhibits non-trivial properties that contradict classical intuition.
Quantum Random Access Codes (QRACs), introduced by Ambainis et al.~\cite{ambainis1999, ambainis2002}, are a concept that embodies the trade-off between the compressibility of information into limited quantum resources and its retrievability.
In an $(n, m)$-QRAC, $n$ bits of classical information are encoded into $m (< n)$ qubits, with the requirement that any arbitrary bit from the original string can be recovered with a probability strictly higher than the classical limit.

QRACs are not merely of theoretical interest; they serve as a critical component in a wide range of areas within quantum communication and computation.
In the field of computational complexity theory, they have been applied to the analysis of communication complexity, the establishment of limits~\cite{klauck2000, aaronson2004}, and the derivation of lower bounds for Locally Decodable Codes~\cite{kerenidis2003}.
Furthermore, from the perspectives of quantum foundations and security, QRACs provide the basis for semi-device-independent protocols such as random number generation and Quantum Key Distribution~\cite{pawlowski2011, li2011}. They are also used in dimension witnesses and self-testing in high-dimensional systems~\cite{brunner2013, tavakoli2021, farkas2019}, as well as in analyzing the performance and verifying the optimality of Mutually Unbiased Bases (MUBs)~\cite{miao2022}.
In these applications, identifying the optimal success probability and the strategy that achieves this bound is crucial for establishing security proofs and performance evaluations of the protocols.

Despite their wide applicability, constructing optimal QRACs for general parameters remains a challenging task.
While analytical bounds are known for high-compression cases such as the $(n, 1)$-QRAC~\cite{ambainis1999, ambainis2002, nayak1999}, determining the optimal encoding states and decoding measurements for general $(n, m)$-QRACs involves complex optimization problems.
Although recent studies have derived analytical upper bounds for general QRACs~\cite{farkas2025}, the specific construction methods to achieve these bounds are not trivial. Indeed, it has been shown that highly symmetric measurements, such as MUBs, are not always optimal~\cite{miao2022}, making the identification of the optimal geometric structure difficult. Consequently, many existing studies have relied on numerical optimization methods such as Semidefinite Programming (SDP) and the See-saw method~\cite{tavakoli2024, navascues2015}.

In contrast to high-compression cases like $(n, 1)$, the success probability in the class of $(n, n-1)$-QRACs, which represents less information loss, asymptotically approaches $1$ as $n \to \infty$.
While deriving optimal constructions for general parameters is extremely difficult, this class exhibits remarkably simple behavior in numerical calculations and is one of the few classes expected to allow for analytical discussion in high-dimensional quantum systems.
In fact, Imamichi et al.~\cite{imamichi2018} performed a heuristic search using SDP for the $(3, 2)$ case and identified the analytical expression of the codes. Furthermore, numerical results strongly suggest that the success probability follows $\mathcal{P} = 1/2 + \sqrt{(n-1)/n}/2$ for general $n$~\cite{mancinska2022}. Although this class has also attracted attention in the context of optimization~\cite{kondo2025}, a concrete code construction that strictly achieves the numerically suggested bound for general $n$ has yet to be discovered.

In this paper, we bridge this gap by providing a general analytical construction for $(n, n-1)$-QRACs.
We establish a construction method using an explicit operator formalism and rigorously demonstrate its validity. Using this method, we prove that the numerically conjectured bound is achievable for all $n$.
This establishes the class as a valuable theoretical model for understanding, without approximation, how quantum advantage is demonstrated for all finite $n$ and how the bound is formed.
Furthermore, from a practical standpoint, we present a specific quantum circuit that decomposes the derived optimal measurements into standard gate operations and computational basis measurements.
The proposed decoding circuit consists solely of interactions between logically adjacent qubits.
Therefore, even on hardware with linear connectivity, such as superconducting qubits, where all-to-all connectivity is not assumed, it can be implemented with a circuit depth of $O(n)$ without the overhead of SWAP gates.
This demonstrates that our construction offers concrete practical utility beyond a mere theoretical proof of existence.
Additionally, based on this construction, we analyze the asymptotic behavior in the high-dimensional limit ($n \to \infty$).
We reveal that while state disturbance caused by the non-commutativity of observables is suppressed by $O(1/n)$, an information-theoretic gap of $O(\log n)$ from the Holevo bound inevitably arises as a physical cost of symmetric encoding using pure states.

\section{Quantum Random Access Codes}\label{sec:qrac}
Alice possesses an $n$-bit classical string $x = (x_1, x_2, \dots, x_n) \in \{0, 1\}^n$, where each bit is chosen independently and uniformly at random.
Alice encodes this data into a quantum state $\rho_x$ using $m$ qubits and sends it to Bob.
To retrieve the $k$-th bit $x_k$, where the index $k \in [n] := \{1, 2, \dots, n\}$ is chosen uniformly at random, Bob performs a measurement described by a POVM $\{M_{0|k}, M_{1|k}\}$ and outputs a guess $b \in \{0, 1\}$.
The optimal average success probability of an $(n,m)$-Quantum Random Access Code ($(n,m)$-QRAC), denoted by $\mathcal{P}_{n,m}^{Q}$, is defined by maximizing over all possible encoding states $\rho_x$ and decoding POVMs $\{M_{b|k}\}$ as follows:
\begin{align}
    \mathcal{P}_{n,m}^{Q} := \max_{\rho_{x}, \{M_{b|k}\}} \frac{1}{n 2^{n}} \sum_{x \in \{0, 1\}^{n}} \sum_{k=1}^{n} \Tr\qty( \rho_x M_{x_{k}|k})
\end{align}

While the exact value of $\mathcal{P}_{n,m}^{Q}$ for general $(n, m)$ remains an open problem, Man\v{c}inska and Storgaard~\cite{mancinska2022} derived the following analytical upper bound:
\begin{align}
    \mathcal{P}_{n,m}^{Q} \leq \frac{1}{2} + \frac{1}{2}\sqrt{\frac{2^{m-1}}{n}}. \label{eq:ms_bound}
\end{align}
However, for the case of $m=n-1$ considered in this paper, this bound exceeds the trivial value of $1$ for $n \ge 3$, and thus provides no non-trivial information constraining the behavior of the average success probability.
On the other hand, the following tighter bound has been conjectured based on numerical calculations~\cite{mancinska2022}:
\begin{align}
    \mathcal{P}_{n,m}^{Q} \stackrel{?}{\leq} \frac{1}{2} + \frac{1}{2}\sqrt{\frac{m}{n}} =: \mathcal{P}_{n,m}^{\mathrm{conj}}. \label{eq:conj_bound}
\end{align}
The main objective of this paper is to construct explicit encoding states and decoding measurements that strictly achieve this conjectured bound $\mathcal{P}_{n,n-1}^{\mathrm{conj}}$ for the $(n, n-1)$ case.

\section{$(n,n-1)$-QRACs}

In this section, we focus on the $(n, n-1)$-QRAC scenario introduced in Sec.~\ref{sec:qrac}.
Specifically, we consider encoding the classical string $x \in \{0, 1\}^n$ into an $(n-1)$-qubit pure state $\ket{\psi_x}$.
Our goal is to explicitly construct the encoding states $\ket{\psi_x}$ and the corresponding decoding POVMs $\{M_{0|k}, M_{1|k}\}$ that satisfy $\mathcal{P}_{n,n-1}^{Q}= \mathcal{P}_{n,n-1}^{\mathrm{conj}}$.

\subsection{Encoding Procedure}

We classify the input bit strings $x \in \{0, 1\}^n$ into the following two sets based on the parity of their Hamming weight $w(x) = \sum_{i=1}^n x_i$:
\begin{align}
    E_n &:= \{ x \in \{0, 1\}^n \mid w(x) \equiv 0 \pmod 2 \}, \\
    O_n &:= \{ x \in \{0, 1\}^n \mid w(x) \equiv 1 \pmod 2 \}.
\end{align}
Note that $|E_n| = |O_n| = 2^{n-1}$. We define the encoding map $x \mapsto \ket{\psi_x}$ into the $m=n-1$ qubit system depending on whether the input $x$ belongs to $E_n$ or $O_n$.

If the input $x$ has even parity ($x \in E_n$), the state $\ket{\psi_x}$ is mapped to the computational basis $\{\ket{z}\}_{z \in \{0,1\}^{n-1}}$.
For $x \in E_n$, since the last bit $x_n$ is uniquely determined once the first $n-1$ bits are fixed, we adopt the first $n-1$ bits as the index for the qubits:
\begin{align}
    \ket{\psi_x} = \ket{x_1 x_2 \dots x_{n-1}} \quad (\text{for } x \in E_n).
\end{align}

If the input $x$ has odd parity ($x \in O_n$), the state $\ket{\psi_x}$ is defined as a superposition of neighboring even parity states ($x' \in E_n$) with a Hamming distance of $1$. In the following, to avoid confusion in notation, we explicitly distinguish elements of the odd parity set as $y \in O_n$ and elements of the even parity set as $x \in E_n$.
\begin{align}
    \ket{\psi_y} = \frac{1}{\sqrt{n}} \sum_{x \in N(y)} \sigma_{yx} \ket{\psi_x} \quad (\text{for } y \in O_n),
\end{align}
where $N(y) = \{x \in E_n \mid d_H(x, y) = 1\}$ denotes the neighborhood set of $y$ ($d_H(x, y)$ represents the Hamming distance between $x$ and $y$).
The coefficients $\sigma_{yx} \in \{+1, -1\}$ are sign functions chosen to satisfy the orthogonality condition $\braket{\psi_y}{\psi_{y'}} = \delta_{yy'}$ between any states within the odd parity group. This specific sign structure $\sigma_{yx}$ can be determined using the following real symmetric matrix $A_n$ by setting $\sigma_{yx} = (A_n)_{yx}$:
\begin{align}
    A_n = \sum_{l=1}^{n} Z^{\otimes (l-1)} \otimes X \otimes I_{n-l}.
\end{align}
The fact that this matrix $A_n$ satisfies the properties necessary for the orthogonality condition (adjacency and $A_n^2 = nI_n$) is proven in Appendix~\ref{app:An_prop}.
By setting $\sigma_{yx} = (A_n)_{yx}$, it is guaranteed that the encoding states $\ket{\psi_y}$ satisfy the orthogonality condition. Indeed,
\begin{align}
    \braket{\psi_y}{\psi_{y'}} &= \frac{1}{n} \sum_{x \in N(y)} \sum_{x' \in N(y')} \sigma_{y x}\sigma_{y' x'} \braket{\psi_{x}}{\psi_{x'}} \\
    &= \frac{1}{n} \sum_{x \in N(y)\cap N(y')} \sigma_{y x}\sigma_{y' x} \\
    &= \frac{1}{n} \sum_{x \in E_n} (A_n)_{yx}(A_n)_{y'x} \\
    &= \delta_{yy'}.
\end{align}

\subsection{Decoding via Optimal POVMs}

We construct the optimal POVM for Bob to recover the $k$-th bit $x_k$ where $k \in [n]$.
First, for each bit value $b \in \{0, 1\}$, we define the following subsets:
\begin{align}
    E_{k, b} &:= \{ x \in E_n \mid x_k = b \}, \\
    O_{k, b} &:= \{ y \in O_n \mid y_k = b \}.
\end{align}
We define the projection operators onto the subspaces corresponding to these sets as $P_{k,b}^{(E)}$ and $P_{k,b}^{(O)}$, respectively.
\begin{align}
    P_{k,b}^{(E)} &:= \sum_{x \in E_{k,b}} \ket{\psi_x}\bra{\psi_x}, \\
    P_{k,b}^{(O)} &:= \sum_{y \in O_{k,b}} \ket{\psi_y}\bra{\psi_y}.
\end{align}

To maximize the average success probability, Bob measures the eigenspace of the following operator corresponding to the maximum eigenvalue:
\begin{align}
    S_{k,b} := P_{k,b}^{(E)} + P_{k,b}^{(O)}.
\end{align}
Regarding the spectral structure of this operator, the following lemma holds.

\begin{lemma}[Spectrum of Sum of Projectors] \label{lem:spectrum_projectors}
    Let $P$ and $Q$ be two projection operators satisfying $P Q P = \mu P$ and $Q P Q = \mu Q$ (where $0 \le \mu \le 1$). Then, the eigenvalues of the sum $S = P + Q$ are $1 \pm \sqrt{\mu}$ and $0$.
\end{lemma}

See Appendix~\ref{app:spectrum} for the proof.
In our construction, the defined pair of projection operators $P=P_{k,b}^{(O)}$ and $Q=P_{k,b}^{(E)}$ satisfy the preconditions of Lemma~\ref{lem:spectrum_projectors}, i.e., $PQP=\mu P$ and $Q P Q = \mu Q$ (see Appendix~\ref{app:projector_condition} for detailed calculations). The parameter $\mu$ is calculated as follows:
\begin{align}
    \mu &= \bra{\psi_y} P_{k, y_k}^{(E)} \ket{\psi_y} \quad (\text{for any } y \in O_{k, b}) \\
    &= \frac{1}{n} \sum_{x \in N(y)} \bra{\psi_x} P_{k, y_k}^{(E)} \ket{\psi_x}.
\end{align}
Note that any element $x$ in $N(y)$ can be uniquely expressed as $x = y \oplus e_l$ (bit-flip at the $l$-th position) for some $l \in [n]$.
In this case, the term $\bra{\psi_{y \oplus e_l}} P_{k, y_k}^{(E)} \ket{\psi_{y \oplus e_l}}$ becomes $1$ only if $l \neq k$, and becomes $0$ if $l=k$ due to orthogonality. Therefore,
\begin{align}
    \mu &= \frac{1}{n} \sum_{l=1}^n \bra{\psi_{y \oplus e_l}} P_{k, y_k}^{(E)} \ket{\psi_{y \oplus e_l}} \\
    &= \frac{1}{n} \sum_{l=1}^n  \delta_{(y \oplus e_l)_k, \, y_k}\\
    &= \frac{n-1}{n}.
\end{align}

Based on this, the spectral decomposition of $S_{k,b}$ can be written as:
\begin{align}
    S_{k,b} = (1+\sqrt{\mu}) \Pi_{k,b}^{(+)} + (1-\sqrt{\mu}) \Pi_{k,b}^{(-)},
\end{align}
where $\Pi_{k,b}^{(\pm)}$ are the projection operators corresponding to the eigenvalues $1 \pm \sqrt{\mu}$.
Using this, we define the POVM element $M_{b|k}$ for detecting bit $b$ as the projection operator corresponding to the maximum eigenvalue:
\begin{align}
    M_{b|k} := \Pi_{k,b}^{(+)}.
\end{align}
Noting the relation
\begin{align}
    \Pi_{k,b}^{(+)} = \frac{1}{2\sqrt{\mu}} \qty( S_{k,b} - (1-\sqrt{\mu})I_{n-1} ) \label{eq:pi_plus}
\end{align}
and the fact that $S_{k,0} + S_{k,1} = 2I_{n-1}$, it follows that
\begin{align}
    M_{0|k} + M_{1|k} = I_{n-1},
\end{align}
confirming that $\{ M_{0|k}, M_{1|k} \}$ constitutes a POVM.

\subsection{Calculation of the Average Success Probability}

We calculate the average success probability $\mathcal{P}^Q_{n, n-1}$ for the $(n, n-1)$-QRACs.
\begin{align}
    \mathcal{P}^Q_{n, n-1} &= \frac{1}{n 2^n} \sum_{x \in \{0, 1\}^n} \sum_{k=1}^n \bra{\psi_x} M_{x_k|k} \ket{\psi_x}.
\end{align}
Using the relations Eq.~\eqref{eq:pi_plus}, we have
\begin{align}
    \bra{\psi_x} M_{x_k|k} \ket{\psi_x} &= \frac{1}{2\sqrt{\mu}} \qty( \expval{S_{k, x_k}}{\psi_x} - (1-\sqrt{\mu}) ).
\end{align}
Note that the expectation value $\expval{S_{k, x_k}}{\psi_x}$ takes a constant value $1+\mu$ for any input $x$ and index $k$ as shown below. By definition, $S_{k, x_k} = P_{k, x_k}^{(E)} + P_{k, x_k}^{(O)}$. The expectation value with respect to the projection operator corresponding to the parity of $x$ (i.e., $P^{(E)}$ if $x \in E_n$, and $P^{(O)}$ if $x \in O_n$) is exactly 1, since the state $\ket{\psi_x}$ lies within that subspace. On the other hand, the expectation value with respect to the other projection operator corresponding to the opposite parity represents the overlap with the neighboring states. Among the $n$ neighbors of $x$, exactly $n-1$ neighbors preserve the value of the $k$-th bit where bit-flips occur at indices $l \neq k$. Each of these neighbors contributes $1/n$ to the overlap. Thus, this term equals $(n-1)/n = \mu$. Consequently, the sum of the expectation values is $1+\mu$ in all cases. Therefore, the average success probability is determined as:
\begin{align}
    \mathcal{P}^Q_{n, n-1} &= \frac{1}{2} (1 + \sqrt{\mu})\\
    &= \frac{1}{2} \qty( 1 + \sqrt{\frac{n-1}{n}} ) = \mathcal{P}_{n,n-1}^{\mathrm{conj}},
\end{align}
yielding a result that matches the conjectured bound Eq.~\eqref{eq:conj_bound}.

\section{Optimal Observables and Decoding Circuits} \label{sec:optimal_obs}
In this section, we explicitly construct the observables corresponding to the optimal POVMs derived in the previous section using Pauli operator strings and discuss their quantum circuit implementation.

\subsection{Structure of Optimal Observables}

The optimal observable $O_k$ for any $k \in [n]$ is defined using the POVM elements as $O_k = M_{0|k} - M_{1|k}$.
Through analysis using the properties of the matrix $A_n$ (see Appendix~\ref{app:derivation_Ok} for the derivation), this can be explicitly described in the following form:

\begin{align}
    O_k &= \sqrt{\frac{n-1}{n}} \mathcal{E}_k + \frac{1}{\sqrt{n(n-1)}} \mathcal{K}_k,
\end{align}
where $\mathcal{E}_k$ represents the diagonal components and $\mathcal{K}_k$ represents the off-diagonal components.
Specifically, they are defined as follows:
\begin{widetext}
\begin{align}
     \mathcal{E}_k &=\begin{dcases}
         Z_k&k=1,\cdots ,n-1\\
         \prod_{j=1}^{n-1} Z_j&k=n
     \end{dcases}\\
     \mathcal{K}_k &=\begin{dcases}
         \sum_{l=k+1}^{n-1}  X_k Z_{k+1} \dots Z_{l-1} X_l+  X_k Z_{k+1} \dots Z_{n-1} +\sum_{l=1}^{k-1}  Y_lZ_{l+1} \dots Z_{k-1} Y_k&k=1,\cdots ,n-1\\
         -\sum_{l=1}^{n-1}   Z_1Z_{2} \dots Z_{l-1}X_l&k=n
     \end{dcases}
\end{align}
\end{widetext}
Note that this $O_k$ can be described as a linear combination of $n$ Pauli strings.

\subsection{Decomposition into Pauli Strings}

In this subsection, we provide a systematic method for constructing a unitary circuit $U$ that transforms the optimal observable $O_k$ into a form $\mathcal{E}_k$ measurable in the computational basis. That is, we construct $U$ satisfying the following condition:
\begin{align}
    U O_k U^\dagger = \mathcal{E}_k
\end{align}

First, we rearrange the observable $O_k$ as a linear combination of $n$ Pauli operator strings $W_j$:
\begin{align}
    O_k = \sum_{j=1}^{n} c_j W_j,
\end{align}
where the Pauli operator string $W_j$ corresponding to index $j$ and the expansion coefficient $c_j$ are defined as follows. First, for the case $k \neq n$:
\begin{itemize}
    \item \textbf{Region I ($1 \le j < k$):} 
    \begin{align}
        W_j = Y_j Z_{j+1} \dots Z_{k-1} Y_k, \quad c_j = \epsilon
    \end{align}
    \item \textbf{Region II ($j = k$):}
    \begin{align}
        W_k = Z_k, \quad c_k = \sqrt{\frac{n-1}{n}}
    \end{align}
    \item \textbf{Region III ($k < j \le n$):}
    \begin{align}
        W_j = \begin{dcases}
            X_k Z_{k+1} \dots Z_{j-1} X_j&j<n\\
            X_k Z_{k+1} \dots Z_{n-1} &j=n
        \end{dcases}, \quad c_j = \epsilon
    \end{align},
\end{itemize}
where $\epsilon := 1/\sqrt{n(n-1)}$, and the coefficients satisfy $\sum c_j^2 = 1$. Next, for the case $k=n$:
\begin{itemize}
    \item \textbf{Region I ($1 \le j < n$):} 
    \begin{align}
        W_j = -  Z_1Z_{2} \dots Z_{j-1}X_j, \quad c_j = \epsilon
    \end{align}
    \item \textbf{Region II ($j = n$):}
    \begin{align}
        W_n = \prod_{j=1}^{n-1} Z_j, \quad c_n = \sqrt{\frac{n-1}{n}}
    \end{align}
\end{itemize}
Note that all $W_j$ are mutually anti-commutative.

\subsection{Circuit Construction}

To transform the operator $O_k$ into $W_k$, we define the unitary operator $U$ as the product of a left unitary operator ($U_L$) and a right unitary operator ($U_R$):
\begin{align}
    U = U_R U_L = (R_k \cdots R_{n-1}) (R_{k-1} \cdots R_1)
\end{align}
The generator $G_m$ and rotation angle $\theta_m$ for the rotation $R_m = \exp\qty( -i \theta_m G_m/2 )$ at each step $m$ are defined as follows:

\begin{itemize}
    \item Left Contraction ($1 \le m < k$): Sequentially transfers coefficients from $W_m$ to $W_{m+1}$, finally transforming into $W_k$.
    \begin{align}
        G_m &= i W_{m+1} W_m \\
        \theta_m &= \begin{dcases}
            \arctan(\sqrt{m}) & (m < k-1) \\
            \arctan\qty( \frac{\sqrt{k-1}\epsilon}{c_k} ) & (m = k-1)
        \end{dcases}
    \end{align}
    \item Right Contraction ($k \le m < n$): Sequentially transfers coefficients from $W_{m+1}$ to $W_m$, finally transforming into $W_k$.
    \begin{align}
        G_m &= i W_m W_{m+1} \\
        \theta_m &= \begin{dcases}
            \arctan(\sqrt{n-m}) & (m > k) \\
            \arctan\qty( \frac{\sqrt{n-k}\epsilon}{\sqrt{(k-1)\epsilon^2 + c_k^2}} ) & (m = k)
        \end{dcases}
    \end{align}
\end{itemize}
Note that the generators $G_m$ involve at most two-body interactions (see Lemma~\ref{lem:locality} in Appendix~\ref{app:locality} for details).

The correctness of the constructed circuit is established by the following theorem.

\begin{theorem}[Diagonalization of Optimal Observable]
    The proposed unitary circuit $U = U_R U_L$ diagonalizes the optimal observable $O_k$ into the observable $\mathcal{E}_k$ measurable in the computational basis:
    \begin{align}
        U O_k U^\dagger = \mathcal{E}_k.
    \end{align}
\end{theorem}

\begin{proof}
    We provide in the rigorous proof using mathematical induction to Appendix~\ref{app:circuit_proof} and provide a sketch of the construction principle in this section.
    
    First, due to the locality of the operators (Lemma~\ref{lem:locality2} in Appendix~\ref{app:locality}), the left transformation $U_L$ and the right transformation $U_R$ act on Pauli strings in independent regions.
    $U_L$ is designed to sequentially transfer the coefficients of the Pauli strings $W_j$ in the left region ($j < k$) to the target $W_k$ through a series of rotations:
    \begin{align}
        U_L \qty( \sum_{j=1}^{k} c_j W_j ) U_L^\dagger = \lambda_L W_k, \quad \text{where } \lambda_L = \sqrt{\sum_{j=1}^k c_j^2}.
    \end{align}
    Similarly, $U_R$ transfers the coefficients from the right region ($j > k$) into $W_k$:
    \begin{align}
        U_R \qty( \lambda_L W_k + \sum_{j=k+1}^{n} c_j W_j ) U_R^\dagger =  W_k.
    \end{align}

    Therefore, the transformed observable becomes $W_k$, i.e., $\mathcal{E}_k$.
\end{proof}

\section{Discussion and Conclusion} \label{sec:discussion_conclusion}

In this study, we established an analytical construction method for the optimal bound of the $(n, n-1)$-QRACs, which had previously remained a subject of numerical conjecture.
In this section, we discuss the asymptotic properties and computational aspects of the derived construction and present the conclusion of this study.

\subsection{Quantum Advantage}

First, we compare the quantum success probability $\mathcal{P}^Q_{n, n-1}$ achieved by this protocol with the classical bound $\mathcal{P}^C_{n, n-1}$.
Using a classical state of $n-1$ bits, the upper bound of the success probability for an $(n, n-1)$-RACs is given by:
\begin{align}
    \mathcal{P}^C_{n, n-1} = \frac{1}{n}\qty(n-\frac{1}{2}).
\end{align}
From this, it can be confirmed that the gap demonstrating quantum advantage, $\Delta\mathcal{P}(n) = \mathcal{P}^Q_{n, n-1} - \mathcal{P}^C_{n, n-1} (>0)$, decreases on the order of $\mathcal{O}(1/n)$.

It is necessary to emphasize the experimental difficulty imposed by this scaling.
The fact that the gap $\Delta\mathcal{P}(n)$ shrinks with the inverse of $n$ implies that the number of measurements required to show a statistically significant difference for quantum advantage scales as $\mathcal{O}(n^2)$.
Furthermore, in actual devices, in addition to errors associated with gate operations, measurement errors occurring during the readout of states exist at non-negligible levels.
Therefore, there is an increasing risk with dimension $n$ that physical noise levels will exceed the minute gap $\Delta\mathcal{P}(n)$, making it difficult to distinguish from the classical bound.

However, this experimental difficulty does not diminish the scientific value of this protocol as a benchmark.
Rather, this difficulty allows the protocol to function as a rigorous dimension witness for the effective dimension.
If experimental results at a certain dimension $n$ fail to significantly exceed the classical bound $P_C(n)$, it implies that the device, even if possessing sufficient physical qubits, is not effectively using the $2^{n-1}$-dimensional Hilbert space as a computational resource due to limitations in noise or control precision.

Conversely, if this rigorous condition is met, it provides strong proof that the device is accurately controlling the vast state space beyond classical capabilities. In this sense, the $(n, n-1)$-QRACs serve as a strict quality assurance metric. Crucially, to make such a rigorous benchmark feasible despite the heavy measurement cost, the implementation cost of the quantum circuit itself must be minimal. This leads us to the operational advantage discussed in the next subsection.

\subsection{Operational Accessibility via $O(n)$ Circuits}

The quantum circuit construction derived in this study demonstrates that quantum advantage in the $(n, n-1)$-QRACs is achievable using only computationally efficient operations.
Generally, as the dimension of the Hilbert space increases, the number of gates required to implement a desired unitary operation scales exponentially. If the optimal encoding and decoding operations required exponential complexity, verifying the quantum advantage would be practically impossible.

However, the decomposition $U = U_R U_L$ we presented reflects the algebraic symmetry and locality of the optimal observables, guaranteeing that the circuit depth and the number of CNOT gates scale with $O(n)$. On the other hand, the circuit depth required for encoding scales with $O(n^2)$ (see Appendix~\ref{app:encode}). Although this cost is higher than the $O(n)$ required for decoding, it remains within polynomial time, allowing for efficient state preparation.
This efficiency confirms that the $(n, n-1)$-QRACs are not merely a theoretical construct but that the quantum advantage of this class can be implemented naturally without requiring complex operations.

\subsection{Transition from Non-commutativity to Commutativity}

Furthermore, our construction highlights a qualitative change in quantum correlations accompanying high-dimensionalization.
In low dimensions (e.g., $n=2, 3$), strong anti-commutativity between observables causes the state to collapse significantly upon a single measurement; however, in high dimensions, the effects of non-commutativity are suppressed.

Specifically, the non-commutativity of the observables $O_k$ we constructed is dominated by the off-diagonal term $\mathcal{K}_k$. Considering that $\mathcal{K}_k$ consists of the sum of approximately $n$ anti-commuting Pauli strings with coefficients of $O(1/n)$, the operator norm of the commutator is evaluated as:
\begin{align}
    \| [O_k, O_l] \| \sim O\qty(\frac{1}{n}) \quad (k \neq l).
\end{align}
This asymptotic commutativity suggests the suppression of state disturbance due to measurement.
Consider the post-measurement state immediately after measuring the first bit $x_1$ and obtaining the correct projective measurement result (projection onto $P_1=M_{x_1|1}$) with probability $\mathcal{P}_1$. The change $\Delta \expval{ O_2 }:= \expval{ O_2 }_{post} - \expval{O_2}_{initial}$ in the expectation value of the observable $O_2$ for the second bit is decomposed into the following two contributions:
\begin{align}
    \Delta\expval{ O_2 } = \underbrace{(\expval{\mathcal{D}}_{post} - \expval{\mathcal{D}}_{initial})}_{\text{(i)}} - \underbrace{\expval{\mathcal{V}}_{initial}}_{\text{(ii)}}.
\end{align}
Introducing the projection operator onto the orthogonal complement $Q_1 := I_{n-1} - P_1$, we define
\begin{align}
    \begin{cases}
        \mathcal{D} := P_1 O_2 P_1 + Q_1 O_2 Q_1, \\
        \mathcal{V} := P_1 O_2 Q_1 + Q_1 O_2 P_1.
    \end{cases}
\end{align}
That is, (i) represents the fluctuation of the diagonal components, and (ii) represents the contribution of the off-diagonal components.
The magnitude of term (ii) is directly limited by the commutator norm between observables:
\begin{align}
    | \text{(ii)} | = | \expval{\mathcal{V}}_{initial} | \leq \frac{1}{2} \| [O_1, O_2] \| \sim O\qty(\frac{1}{n}).
\end{align}
On the other hand, term (i) is proportional to the failure probability of the first bit measurement. Using the main result of this study, $\mathcal{P}_1=\mathcal{P}^Q_{n, n-1}  = 1 - O(1/n)$, we obtain:
\begin{align}
    | \text{(i)} | \leq 2 \|\mathcal{D}\| (1-\mathcal{P}_1) \sim O\qty(\frac{1}{n}).
\end{align}
In conclusion, $\Delta\expval{ O_2 }$ is suppressed on the order of $O(1/n)$.
This theoretically guarantees that in the limit $n \to \infty$, the influence of measuring one bit on the information of other bits vanishes, enabling continuous readout.

However, the attempt to identify all bits simultaneously is regulated by the limitations of the quantum system's dimension.
Since the dimension of the Hilbert space used for encoding is $d=2^{n-1}$, the maximum number of mutually orthogonal and perfectly distinguishable states within this space is limited to $2^{n-1}$.
In contrast, the total number of input bit strings to be identified is $2^n$.
Therefore, even with any measurement strategy, the average probability of correctly identifying the entire bit string without error cannot exceed the dimensional ratio of $1/2$.
That is, while the readout accuracy for individual bits approaches $1$ with high-dimension, the complete extraction of all information is limited by the dimension, which is consistent with the derived scaling $\Delta\expval{ O_2 } \sim O(1/n)$.

\subsection{Information Theoretic Cost of Symmetry}

Finally, to evaluate the information-theoretic efficiency of this protocol, we analyze the gap $\Delta \mathcal{I}(n)$ between the Holevo bound, which is the fundamental limit of quantum information capacity, and the total mutual information obtained by this method.

First, we define the difference between the Holevo bound and the total mutual information $\mathcal{I}_{\text{total}}(n)$ obtained in this protocol as follows:
\begin{align}
    \Delta \mathcal{I}(n) := (n-1) - \mathcal{I}_{\text{total}}(n).
\end{align}
Note that due to the symmetry of the encoding and the independent identical distribution of each bit, the total mutual information is $n$ times the mutual information of a single bit. Using the binary entropy function $H(p) := -p \log_2 p - (1-p) \log_2 (1-p)$, this can be written as:
\begin{align}
    \mathcal{I}_{\text{total}}(n)
    = n(1 - H(\mathcal{P}^Q_{n, n-1})).
\end{align}
Therefore, the gap $\Delta \mathcal{I}(n)$ is expressed as a function of the success probability $\mathcal{P}^Q_{n, n-1} $ as:
\begin{align}
    \Delta \mathcal{I}(n) = n H(\mathcal{P}^Q_{n, n-1}) - 1.
\end{align}

We evaluate the asymptotic behavior of $\Delta \mathcal{I}(n)$ as $n \to \infty$.
From the results of the previous section, the success probability is given by $\mathcal{P}^Q_{n, n-1} \approx 1 - 1/(4n)$. Thus,
\begin{align}
    H(\mathcal{P}^Q_{n, n-1}) &\approx \frac{2 + \log_2 n}{4n}.
\end{align}
Substituting this into the definition of the gap yields the following scaling law:
\begin{align}
    \Delta \mathcal{I}(n) \sim O(\log n).
\end{align}
This result indicates that the information loss relative to the Holevo bound $n-1$ does not converge to a constant but increases logarithmically with the dimension $n$.
This reflects the geometric constraints of this protocol.

If one were to compress $n$ bits into $n-1$ bits using a classical strategy, simply adopting the strategy of ``discarding 1 bit and sending the remaining $n-1$ bits directly" would result in a total information sum of $n-1$, achieving the Holevo bound. However, this strategy possesses a strong asymmetry in that the information of the discarded bit is completely lost.

In contrast, the quantum protocol in this study satisfies the strict constraint of treating all bits fairly.
The derived gap of $O(\log n)$ can be interpreted as the cost that must be paid to embed information fairly into a quantum state.

\subsection{Conclusion}

In conclusion, this study provides a rigorous analytical foundation for $(n, n-1)$-QRACs, resolving the conjecture regarding their optimal success probability. By establishing an explicit operator formalism, we have bridged the gap between numerical observations and theoretical proof. Furthermore, we translated this algebraic framework into a practical gate decomposition that yields an $O(n)$-depth decoding circuit, ensuring hardware-efficient implementation under linear connectivity constraints. Our asymptotic analysis further clarifies the quantum-to-classical transition: while the non-commutativity of measurements vanishes in the high-dimensional limit, a logarithmic information deficit relative to the Holevo bound persists. These results offer deeper insights into the theoretic constraints at the boundary of quantum and classical regimes.

\appendix

\section{Properties of the Matrix $A_n$} \label{app:An_prop}

\begin{lemma}[Properties of $A_n$] \label{lem:An_properties}
    For any $n \ge 1$, the matrix
    \begin{align}
        A_n = \sum_{l=1}^{n} Z^{\otimes (l-1)} \otimes X \otimes I_{n-l}
    \end{align}
    satisfies the following properties:
    \begin{enumerate}
        \item \textbf{Adjacency:} The matrix element $(A_n)_{yx}$ is non-zero if and only if the Hamming distance is $d_H(y, x) = 1$.
        \item \textbf{Values:} The values of the non-zero elements are either $\pm 1$. That is, $(A_n)_{yx} \in \{0, 1, -1\}$.
        \item \textbf{Orthogonality:} $A_n^2 = n I_{n}$.
    \end{enumerate}
\end{lemma}

\begin{proof}
    We prove this by mathematical induction.
    
    First, for the case $n=1$, $A_1 = \mqty(0 & 1 \\ 1 & 0)$.
    The non-zero elements are only the $(0,1)$ and $(1,0)$ components, and their Hamming distance is $1$. All non-zero entries are equal to $1$. Moreover, since $A_1^2 = I_1$, all properties hold.

    Next, assume that the properties hold for $n-1$. Based on the block structure of $A_n$, we decompose the indices $y, x \in \{0, 1\}^n$ into the first bit $y_1, x_1$ and the remaining $n-1$ bits $y', x'$.
    From
    \begin{align} 
    A_n &= X \otimes I_{n-1}+ Z \otimes A_{n-1},
    \end{align}
    the elements satisfy:
    \begin{align}
        (A_n)_{yx} = \begin{cases}
            (A_{n-1})_{y'x'} & \text{if } y_1 = x_1 = 0 \\
            -(A_{n-1})_{y'x'} & \text{if } y_1 = x_1 = 1 \\
            \delta_{y'x'} & \text{if } y_1 \neq x_1
        \end{cases}.
    \end{align}
    
    \textbf{1. Adjacency \& 2. Values:}
    \begin{itemize}
        \item Case $y_1 = x_1$: $(A_n)_{yx} = \pm (A_{n-1})_{y'x'}$. By the inductive hypothesis, this is non-zero only when $d_H(y', x') = 1$. In this case, the total Hamming distance is $d_H(y, x) = 1$, and the value is $\pm 1$.
        \item Case $y_1 \neq x_1$: $(A_n)_{yx} 
        = \delta_{y'x'}$. This is non-zero only when $y' = x'$, i.e., $d_H(y', x') = 0$. In this case, the total Hamming distance is $d_H(y, x) = 1$, and the value is $1$.
    \end{itemize}
    From the above, in either case, the non-zero condition coincides with $d_H(y, x) = 1$, and the values are $\pm 1$.

    \textbf{3. Orthogonality:}
    Computing the square of the matrix yields:
    \begin{align}
        A_n^2 
        &= \mqty(A_{n-1}^2 + I_{n-1} & A_{n-1} - A_{n-1} \\ A_{n-1} - A_{n-1} & I_{n-1} + A_{n-1}^2) \\
        &= \mqty( (n-1)I_{n-1} + I_{n-1} & 0 \\ 0 & (n-1)I_{n-1} + I_{n-1} ) = n I_{n}.
    \end{align}
    Therefore, the properties hold for any $n$.
\end{proof}

\begin{lemma}[Matrix Elements of $A_n$] \label{lem:An_transition}
    The matrix element $(A_n)_{yx} = \bra{y} A_n \ket{x}$ of the matrix $A_n$ in an $n$-qubit system is non-zero only when the Hamming distance is $d_H(x, y) = 1$, where $\ket{x}=\ket{x_1\dots x_n}$ and $\ket{y}=\ket{y_1\dots y_n}$. Let $l$ be the position of the flipped bit. Then, the element is given by:
    \begin{align}
        (A_n)_{yx} = \prod_{s=1}^{l-1} (-1)^{x_s}.
    \end{align}
\end{lemma}

\begin{proof}
    We calculate this based on the definition $A_n = \sum_{l=1}^{n} g_l$ (where $g_l = Z^{\otimes (l-1)} \otimes X \otimes I_{n-l}$).
    Applying $g_l$ to the basis state $\ket{x}$, we get:
    \begin{align}
        g_l \ket{x_1 \dots x_n} &= (Z \ket{x_1}) \otimes \dots \otimes (Z \ket{x_{l-1}}) \otimes (X \ket{x_l}) \otimes \dots \\
        &= (-1)^{x_1} \dots (-1)^{x_{l-1}} \ket{x_1 \dots \bar{x}_l \dots x_n} \\
        &= \qty( \prod_{s=1}^{l-1} (-1)^{x_s} ) \ket{x \oplus e_l}.
    \end{align}
    Therefore, $\bra{y} A_n \ket{x}$ is non-zero only when $y = x \oplus e_l$, and it coincides with the phase factor derived above.
\end{proof}

\section{Proof of Lemma~\ref{lem:spectrum_projectors}} \label{app:spectrum}

\begin{proof}
    Let $\lambda$ be an eigenvalue of $S$ and $\ket{\lambda}$ be the corresponding eigenket. By definition,
    \begin{align}
        (P + Q) \ket{\lambda} = \lambda \ket{\lambda}
    \end{align}
    holds. Applying $P$ from the left to this equation and using $P^2 = P$, we obtain:
    \begin{align}
        P \ket{\lambda} + P Q \ket{\lambda} &= \lambda P \ket{\lambda} \\
        P Q \ket{\lambda} &= (\lambda - 1) P \ket{\lambda}. \label{eq:PQ_lambda}
    \end{align}
    Similarly, applying $Q$ from the left to the original equation and using $Q^2 = Q$, we obtain:
    \begin{align}
        Q P \ket{\lambda} + Q \ket{\lambda} &= \lambda Q \ket{\lambda}\\
        Q P \ket{\lambda} &= (\lambda - 1) Q \ket{\lambda}. \label{eq:QP_lambda}
    \end{align}
    
    Now, applying $P$ from the left to Eq.~\eqref{eq:QP_lambda}, the left-hand side can be transformed using the assumption $P Q P = \mu P$, while the right-hand side can be transformed by substituting Eq.~\eqref{eq:PQ_lambda}:
    \begin{align}
        P (Q P \ket{\lambda}) &= (\lambda - 1) P Q \ket{\lambda} \\
        (P Q P) \ket{\lambda} &= (\lambda - 1) (\lambda - 1) P \ket{\lambda} \\
        \mu P \ket{\lambda} &= (\lambda - 1)^2 P \ket{\lambda}.
    \end{align}
    From this, we obtain the following relation:
    \begin{align}
        \qty( (\lambda - 1)^2 - \mu ) P \ket{\lambda} = 0. \label{eq:eigen_condition}
    \end{align}
    
    We now consider two cases:
    \begin{enumerate}
        \item Case where $P \ket{\lambda} \neq 0$:\\
        Since the coefficient must be zero according to Eq.~\eqref{eq:eigen_condition}, we have
        \begin{align}
            (\lambda - 1)^2 = \mu \quad \Longrightarrow \quad \lambda = 1 \pm \sqrt{\mu}.
        \end{align}
        
        \item Case where $P \ket{\lambda} = 0$:\\
        Substituting $P\ket{\lambda}=0$ into the original equation yields $Q\ket{\lambda} = \lambda \ket{\lambda}$. Since $Q$ is a projection operator, it follows that $\lambda \in \{0, 1\}$.
    \end{enumerate}
    In conclusion, the spectrum of $S$ is $\{1 \pm \sqrt{\mu}, 0, 1\}$.
\end{proof}

\section{Verification of Algebraic Condition for Projectors} \label{app:projector_condition}

In this appendix, we prove that the defined projection operators $P=P_{k,b}^{(O)}$ and $Q=P_{k,b}^{(E)}$ satisfy $PQP = \mu P$ and $QPQ = \mu Q$.
Since $P$ is the projection operator onto the subspace spanned by the orthonormal basis $\{ \ket{\psi_y} \}_{y \in O_{k,b}}$, this is equivalent to the following condition on the matrix elements:
\begin{align}
    \bra{\psi_y} Q \ket{\psi_{y'}} = \mu \delta_{yy'} \quad (\forall y, y' \in O_{k,b}).
\end{align}

Expanding the left-hand side, we obtain:
\begin{widetext}
\begin{align}
    \bra{\psi_y} Q \ket{\psi_{y'}} &= \sum_{x \in E_{k,b}} \braket{\psi_y}{\psi_x}\braket{\psi_x}{\psi_{y'}} \\
    &= \frac{1}{n} \sum_{x \in E_{k,b}} \sum_{\tilde{x} \in N(y)} (A_n)_{y\tilde{x}} \braket{\psi_{\tilde{x}}}{\psi_x} \sum_{\tilde{x}' \in N(y')} (A_n)_{y'\tilde{x}'} \braket{\psi_x}{\psi_{\tilde{x}'}}\\
    &= \frac{1}{n} \sum_{x \in E_{k,b} \cap N(y) \cap N(y')} (A_n)_{yx} (A_n)_{y'x}.
\end{align}    
\end{widetext}

Note that $y, y' \in O_{k,b}$, which implies $y_k = y'_k = b$.
Since a neighbor $x$ is a 1-bit flip of $y$, if $x \in N(y)$, it can be written as $x = y \oplus e_l$ for some $l \in [n]$.
Furthermore, for the condition $x \in E_{k,b}$ (i.e., $x_k = b$) to hold, the flipped position $l$ must not be $k$.
Therefore, the summation is restricted to indices $l \neq k$.

\paragraph*{Case 1: $y = y'$}
In this case,
\begin{align}
    \bra{\psi_y} Q \ket{\psi_y} &= \frac{1}{n} \sum_{l \neq k} (A_n)_{y, y \oplus e_l}^2 \\
    &= \frac{1}{n} \sum_{l \neq k} (\pm 1)^2 = \frac{n-1}{n}.
\end{align}
This yields $\mu = (n-1)/n$.

\paragraph*{Case 2: $y \neq y'$}
Since both $y$ and $y'$ belong to $O_{k,b}$, the Hamming distance $d_H(y, y')$ is even and at least $2$.
\begin{itemize}
    \item Case $d_H(y, y') \ge 4$: Since there are no common neighbors, the sum is $0$.
    \item Case $d_H(y, y') = 2$: There are exactly two common neighbors in $N(y) \cap N(y')$. Let these be $u$ and $v$.
    Since the $k$-th bits of $y$ and $y'$ are equal, their difference originates from two bit positions $p, q$ other than $k$ (i.e., $p, q \neq k$). Consequently, the $k$-th bit of the common neighbors $u, v$ also remains $b$, satisfying $u, v \in E_{k,b}$.
    Hence, this sum coincides with the off-diagonal components of $A_n^2$:
    \begin{align}
        &\sum_{x \in E_{k,b} \cap N(y) \cap N(y')} (A_n)_{yx} (A_n)_{y'x} \\
        &= (A_n)_{yu}(A_n)_{y'u} + (A_n)_{yv}(A_n)_{y'v} \\
        &= (A_n^2)_{yy'}.
    \end{align}
    From Lemma~\ref{lem:An_properties} ($A_n^2 = nI_n$), we have $(A_n^2)_{yy'} = 0$ for distinct $y, y'$. Thus, the off-diagonal terms vanish.
\end{itemize}

Thus, we conclude that $PQP = \frac{n-1}{n} P$. By symmetry, $QPQ = \frac{n-1}{n} Q$ holds similarly.

\section{Derivation of Optimal Observables} \label{app:derivation_Ok}

In this appendix, we derive the explicit form of the optimal measurement observable $O_k$.
By definition,
\begin{align}
    O_k &= M_{0|k} - M_{1|k} = \Pi_{0,k}^{(+)} - \Pi_{1,k}^{(+)} \\
    &= \frac{1}{2\sqrt{\mu}} \Delta_k
\end{align}
(where $\Delta_k = S_{k,0} - S_{k,1}$). The operator $\Delta_k$ can be expanded as follows:
\begin{align}
    \Delta_k &= \sum_{y \in O_n} (-1)^{y_k} \ket{\psi_y}\bra{\psi_y} + \sum_{x \in E_n} (-1)^{x_k} \ket{\psi_x}\bra{\psi_x}.
\end{align}

Let the second term be $\mathcal{E}_k$.
From the definition of encoding, $\ket{\psi_x} = \ket{x_1 \dots x_{n-1}}$ for any $x \in E_n$. Furthermore, due to the condition $x \in E_n$, once the first $n-1$ bits $x' = (x_1, \dots, x_{n-1})$ are determined, the last bit $x_n$ is uniquely fixed. Therefore, the summation is equivalent to the sum over all patterns of $x'$, i.e., $\sum_{x' \in \{0,1\}^{n-1}}$.

Here, we consider cases based on the value of $k$.

\paragraph*{Case 1: $k \in [n-1]$}
In this case, since $(-1)^{x_k}$ is the sign corresponding to the qubit index $k$, we have:
\begin{align}
    \mathcal{E}_k &= \sum_{x' \in \{0,1\}^{n-1}} (-1)^{x'_k} \ket{x'}\bra{x'} \\
    &= Z_k.
\end{align}

\paragraph*{Case 2: $k = n$}
In this case, the sign is $(-1)^{x_n}$. Since $\sum_{j=1}^n x_j \equiv 0 \pmod 2$ for $x \in E_n$, it holds that
\begin{align}
    x_n \equiv \sum_{j=1}^{n-1} x_j \pmod 2.
\end{align}
Thus, the sign factor can be transformed as
\begin{align}
    (-1)^{x_n} = (-1)^{\sum_{j=1}^{n-1} x_j} = \prod_{j=1}^{n-1} (-1)^{x_j}.
\end{align}
Using this, we obtain:
\begin{align}
    \mathcal{E}_n &= \sum_{x' \in \{0,1\}^{n-1}} \qty( \prod_{j=1}^{n-1} (-1)^{x'_j} ) \ket{x'}\bra{x'} \\
    &= \sum_{x_1 \dots x_{n-1}} \qty( (-1)^{x_1}\ket{x_1}\bra{x_1} \otimes \dots \otimes (-1)^{x_{n-1}}\ket{x_{n-1}}\bra{x_{n-1}} ) \\
    &= \bigotimes_{j=1}^{n-1} \qty( \sum_{b \in \{0,1\}} (-1)^b \ket{b}\bra{b} ) \\
    &= Z^{\otimes (n-1)},
\end{align}
which is the parity operator for all qubits.

On the other hand, the first term is calculated using the coefficient matrix $A_n$ as follows:
\begin{widetext}
\begin{align}
    \sum_{y \in O_n} (-1)^{y_k} \ket{\psi_y}\bra{\psi_y} 
    &= \frac{1}{n} \sum_{y \in O_n} (-1)^{y_k} \sum_{x, x' \in N(y)} (A_n)_{yx} (A_n)_{yx'} \ket{\psi_x}\bra{\psi_{x'}} \\
    &= \underbrace{\frac{1}{n} \sum_{y \in O_n} \sum_{x \in N(y)} (-1)^{y_k} ((A_n)_{yx})^2 \ket{\psi_x}\bra{\psi_x}}_{\text{Diagonal Part}} 
    + \underbrace{\frac{1}{n} \sum_{y \in O_n} \sum_{\substack{x, x' \in N(y) \\ x \neq x'}} (-1)^{y_k} (A_n)_{yx} (A_n)_{yx'} \ket{\psi_x}\bra{\psi_{x'}}}_{\text{Off-diagonal Part}}.
\end{align}    
\end{widetext}

\subsubsection*{Diagonal Part}
We calculate the contribution of the diagonal components.
\begin{align}
    (\text{Diag}) &= \frac{1}{n} \sum_{x \in E_n} \ket{\psi_x}\bra{\psi_x} \sum_{y \in N(x)} (-1)^{y_k}.
\end{align}
Consider the inner sum $\sum_{y \in N(x)} (-1)^{y_k}$.
Each element $y$ in the neighborhood $N(x)$ can be expressed as $y = x \oplus e_j$ for some bit position $j \in [n]$.
For the index $k$ of interest:
\begin{itemize}
    \item If $j = k$: Since $y_k \ne x_k$, the sign is $(-1)^{y_k} = -(-1)^{x_k}$.
    \item If $j \ne k$: Since $y_k = x_k$, the sign is $(-1)^{y_k} = (-1)^{x_k}$.
\end{itemize}
This logic holds for all $k$, including $k=n$. Therefore,
\begin{align}
    \sum_{y \in N(x)} (-1)^{y_k} &= 1 \cdot (-(-1)^{x_k}) + (n-1) \cdot (-1)^{x_k} \\
    &= (n-2)(-1)^{x_k}.
\end{align}
Substituting this back into the original expression, we obtain:
\begin{align}
    (\text{Diag}) &= \frac{n-2}{n} \sum_{x \in E_n} (-1)^{x_k} \ket{\psi_x}\bra{\psi_x} \\
    &= \frac{n-2}{n} \mathcal{E}_k,
\end{align}
where $\mathcal{E}_k$ is the operator for the even parity part mentioned above, where $\mathcal{E}_k = Z_k$ or $Z^{\otimes (n-1)}$ when $k=n$.

\subsubsection*{Off-diagonal Part}
We calculate the contribution of the off-diagonal components where $x \neq x'$.
\begin{align}
    (\text{Off}) &= \frac{1}{n} \sum_{\substack{x, x' \in E_n \\ d(x,x')=2}} \ket{\psi_x}\bra{\psi_{x'}} \underbrace{ \sum_{y \in N(x) \cap N(x')} (-1)^{y_k} (A_n)_{yx} (A_n)_{yx'} }_{C_{x,x'}}.
\end{align}
Consider the coefficient part $C_{x,x'}$. Since the Hamming distance between $x$ and $x'$ is 2, there are exactly two common neighbors, denoted by $u$ and $v$.

Let the differing bit positions between $x$ and $x'$ be $j$ and $m$ (i.e., $x' = x \oplus e_j \oplus e_m$). Then the common neighbors are $u = x \oplus e_j$ and $v = x \oplus e_m$.
If neither $j$ nor $m$ is equal to the observation index $k$, then the $k$-th bit is not flipped in the neighbors $u, v$, so $u_k = v_k = x_k$. In this case, the sign factor is common, and we have:
\begin{align}
    C_{x,x'} = (-1)^{x_k} \qty( (A_n)_{ux}(A_n)_{ux'} + (A_n)_{vx}(A_n)_{vx'} ).
\end{align}
However, due to the property of the matrix $A_n$ ($A_n^2 = nI_n$), the sum over all neighbors for distinct $x, x'$ is $\sum_y (A_n)_{yx}(A_n)_{yx'} = 0$. In the case $j, m \neq k$, since there are no contributions from paths other than these two neighbors $u, v$, this partial sum itself becomes $0$ and vanishes.

Therefore, non-zero contributions remain only when one of the flipped positions is $k$.
Let the other flipped position be $l (\neq k)$. That is, the two neighbors are:
\begin{itemize}
    \item $u = x \oplus e_k \implies (-1)^{u_k} = -(-1)^{x_k}$
    \item $v = x \oplus e_l \implies (-1)^{v_k} = (-1)^{x_k}$
\end{itemize}
In this case,
\begin{align}
    C_{x,x'} &= (-1)^{u_k}(A_n)_{ux}(A_n)_{ux'} + (-1)^{v_k}(A_n)_{vx}(A_n)_{vx'} \\
    &= (-1)^{x_k} \qty( - (A_n)_{ux}(A_n)_{ux'} + (A_n)_{vx}(A_n)_{vx'} ). \label{eq:coeff_diff}
\end{align}

We now consider cases based on the relative order of $k$ and $l$.

\paragraph*{Case 1: Terms with $l > k$}
From Lemma~\ref{lem:An_transition} in Appendix~\ref{app:An_prop}, the matrix element $(A_n)_{yx} = \bra{y} A_n \ket{x}$ gives the following phase factor when the flipped bit position is $p$:
\begin{align}
    (A_n)_{yx} = \prod_{s=1}^{p-1} (-1)^{x_s}.
\end{align}
From the above,
\begin{align}
    C_{x,x'} &= (-1)^{x_k} \qty( -  \prod_{s=1}^{k-1} (-1)^{x_s}\prod_{s'=1}^{l-1} (-1)^{x'_{s'}} + \prod_{s=1}^{l-1} (-1)^{x_s}\prod_{s'=1}^{k-1} (-1)^{x'_{s'}}  )\\
    &= 2 \prod_{s=k+1}^{l-1} (-1)^{x_{s}}.
\end{align}

\paragraph*{Case 2: Terms with $l < k$}
In this case,
\begin{align}
    C_{x,x'} &= (-1)^{x_k} \qty( -  \prod_{s=1}^{k-1} (-1)^{x_s}\prod_{s'=1}^{l-1} (-1)^{x'_{s'}} + \prod_{s=1}^{l-1} (-1)^{x_s}\prod_{s'=1}^{k-1} (-1)^{x'_{s'}}  )\\
    &= -2 \prod_{s=l}^{k} (-1)^{x_{s}}.
\end{align}

\begin{widetext}
Based on the above,
\begin{align}
    (\text{Off}) &= \frac{1}{n} \sum_{x \in E_n} \sum_{l \neq k} \ket{\psi_x}\bra{\psi_{x'}} C_{x,x'} \quad (\text{where } x' = x \oplus e_k \oplus e_l) \\
    &= \frac{1}{n} \qty( \sum_{x \in E_n} \sum_{l > k} \ket{\psi_x}\bra{\psi_{x'}} C_{x,x'} + \sum_{x \in E_n} \sum_{l < k} \ket{\psi_x}\bra{\psi_{x'}} C_{x,x'} )\\
    &=\frac{2}{n}\left(\sum_{l=k+1}^{n-1}  X_k Z_{k+1} \dots Z_{l-1} X_l+(1-\delta_{kn})  X_k Z_{k+1} \dots Z_{n-1} \right.\\
    &\quad \left.-\sum_{l=1}^{k-1}   Z_lX_l Z_{l+1} \dots Z_kX_k-\sum_{x \in E_n} \sum_{l < n}  \ket{\psi_x}\bra{\psi_{x'}}  (-1)^{x_{n}}\prod_{s=l}^{n-1} (-1)^{x_{s}}\right)\\
    &=\frac{2}{n}\qty(\sum_{l=k+1}^{n-1}  X_k Z_{k+1} \dots Z_{l-1} X_l+(1-\delta_{kn})  X_k Z_{k+1} \dots Z_{n-1} +\sum_{l=1}^{k-1}   Y_lZ_{l+1} \dots Z_{k-1} Y_k-\delta_{kn}\sum_{x \in E_n} \sum_{l < n} \ket{\psi_x}\bra{\psi_{x'}} \prod_{s=1}^{l-1} (-1)^{x_{s}})\\
    &=\frac{2}{n}\qty(\sum_{l=k+1}^{n-1}  X_k Z_{k+1} \dots Z_{l-1} X_l+(1-\delta_{kn})  X_k Z_{k+1} \dots Z_{n-1} +\sum_{l=1}^{k-1}   Y_lZ_{l+1} \dots Z_{k-1} Y_k-\delta_{kn}\sum_{l=1}^{n-1}   Z_1Z_{2} \dots Z_{l-1}X_l)\\
    &=\frac{2}{n}\mathcal{K}_k.
\end{align}
We thus obtain the desired result.    
\end{widetext}

\section{Locality of Generators} \label{app:locality}

\begin{lemma}\label{lem:locality}
    $G_m = i W_{m+1} W_m$ involves at most two-body interactions.
\end{lemma}

\begin{proof}

\

\subsection*{Case $k \neq n$}

\subsubsection*{Case 1: $m < k-1$}

Both $W_m$ and $W_{m+1}$ take the form of Region I.
\begin{align}
    W_m &= Y_m Z_{m+1} (Z_{m+2} \dots Z_{k-1} Y_k) \\
    W_{m+1} &=  Y_{m+1} (Z_{m+2} \dots Z_{k-1} Y_k)
\end{align}
Thus,
\begin{align}
    W_{m+1} W_m\\
    &= Y_m (Y_{m+1} Z_{m+1}) \\
    &= Y_m (i X_{m+1}).
\end{align}
Therefore,
\begin{align}
    G_m = i (i Y_m X_{m+1}) = - Y_m X_{m+1}.
\end{align}
This is a two-body operator acting only on sites $m$ and $m+1$.

\subsubsection*{Case 2: $m = k-1$}

Consider the product of $W_{k-1}$ (Region I) and $W_k$ (Region II).
\begin{align}
    W_{k-1} &= Y_{k-1} Y_k \\
    W_k &= Z_k
\end{align}
Thus,
\begin{align}
    W_k W_{k-1}\\
    &= Y_{k-1} (Z_k Y_k) \\
    &= Y_{k-1} (-i X_k).
\end{align}
Therefore,
\begin{align}
    G_{k-1} = i W_k W_{k-1} = i (-i Y_{k-1} X_k) = Y_{k-1} X_k.
\end{align}
This is a two-body operator acting on sites $k-1$ and $k$.

\subsubsection*{Case 3: $m = k$}

Consider the product of $W_k$ (Region II) and $W_{k+1}$ (Region III).
\begin{align}
    W_k &= Z_k \\
    W_{k+1} &= X_k X_{k+1}
\end{align}
Thus,
\begin{align}
    W_{k+1} W_k  \\
    &= (X_k Z_k) X_{k+1} \\
    &= (-i Y_k) X_{k+1}.
\end{align}
Therefore,
\begin{align}
    G_k = i W_{k+1} W_k = i (-i Y_k X_{k+1}) = Y_k X_{k+1}.
\end{align}
This is a two-body operator acting on sites $k$ and $k+1$.

\subsubsection*{Case 4: $m > k$}

Both $W_m$ and $W_{m+1}$ take the form of Region III.
\begin{align}
    W_m &= (X_k Z_{k+1} \dots Z_{m-1}) X_m \\
    W_{m+1} &= (X_k Z_{k+1} \dots Z_{m-1}) Z_m X_{m+1}
\end{align}
Thus,
\begin{align}
    W_{m+1} W_m  \\
    &= (Z_m X_m) X_{m+1} \\
    &= (i Y_m) X_{m+1}.
\end{align}
Therefore,
\begin{align}
    G_m = i (i Y_m X_{m+1}) = - Y_m X_{m+1}.
\end{align}
This is a two-body operator acting only on sites $m$ and $m+1$.

Moreover, for the specific boundary case:
\begin{align}
    W_{n-1} W_{n-2} 
    &= Z_{n-1}Z_{n-2}X_{n-2} \\
    &= (i Y_{n-2}) Z_{n-1}.
\end{align}
Therefore,
\begin{align}
    G_{n-2} = i (i Y_{n-2} Z_{n-1}) = - Y_{n-2} Z_{n-1}.
\end{align}
This is a two-body operator acting only on sites $n-2$ and $n-1$.

\subsection*{Case $k = n$}

\subsubsection*{Case 1: $m < n-1$}

Both $W_m$ and $W_{m+1}$ take the form of Region I.
\begin{align}
    W_m &=  -   Z_1Z_{2} \dots Z_{m-1}X_m \\
    W_{m+1} &=   - Z_1Z_{2} \dots Z_{m}X_{m+1}
\end{align}
Thus,
\begin{align}
    W_{m+1} W_m
    &= Z_m X_m X_{m+1} \\
    &= i Y_m X_{m+1}.
\end{align}
Therefore,
\begin{align}
    G_m = i (i Y_m X_{m+1}) = - Y_m X_{m+1}.
\end{align}
This is a two-body operator acting only on sites $m$ and $m+1$.

\subsubsection*{Case 2: $m = n-1$}

Consider the product of $W_{n-1}$ (Region I) and $W_n$ (Region II).
\begin{align}
    W_{n-1} &=  - Z_1Z_{2} \dots Z_{n-2}X_{n-1} \\
    W_n &= Z_1 \dots Z_{n-1}
\end{align}
Thus,
\begin{align}
    W_n W_{n-1}
    &=- Z_{n-1} X_{n-1} \\
    &= -i Y_{n-1}.
\end{align}
Therefore,
\begin{align}
    G_{n-1} = i W_n W_{n-1} =  Y_{n-1}.
\end{align}
This is a one-body operator acting on site $n-1$.

In all cases above ($1 \le m < n$), it has been shown that $G_m = i W_{m+1} W_m$ is a product of Pauli operators acting on at most two adjacent sites ($m, m+1$).

\end{proof}

\begin{lemma}[Locality of Operators]\label{lem:locality2}
    The left-side rotation $R_m \ (m < k)$ does not affect the terms $W_j \ (j > m+1)$ located in the right-side region. Similarly, the right-side rotations do not affect the left-side region.
\end{lemma}
\begin{proof}
    The generator for the left-side rotation is $G_m = i W_{m+1} W_m$.
    For any $j > m+1$, since $W_j$ anti-commutes with both $W_m$ and $W_{m+1}$, we have:
    \begin{align}
        [G_m, W_j] &= i [W_{m+1} W_m, W_j] \\
        &= i ( W_{m+1} W_m W_j - W_j W_{m+1} W_m ) = 0.
    \end{align}
    Thus, they commute. Therefore,
    \begin{align}
        R_m W_j R_m^\dagger = W_j \quad (\text{for } j > m+1).
    \end{align}
    This ensures that $U_L$ does not act on $W_{k+1}, \dots, W_n$, and $U_R$ does not act on $W_1, \dots, W_{k-1}$.
\end{proof}

\section{Proof of Circuit Correctness} \label{app:circuit_proof}

\begin{proof}
We show that the constructed $U$ satisfies $U O_k U^\dagger = W_k$.
First, from Lemma~\ref{lem:locality2}, $U_L$ does not act on Region III ($j > k$), and $U_R$ does not act on Region I ($j < k$). Therefore, we can consider the left and right contractions independently.
Furthermore, for two mutually anti-commuting operators $V_1, V_2$ that square to the identity, considering the rotation $R = \exp(-i \theta G/2)$ with the generator $G = i V_1 V_2$, we use the following relations throughout this proof:
\begin{align}
    R V_1 R^\dagger &= V_1 \cos\theta + V_2 \sin\theta \\
    R V_2 R^\dagger &= -V_1 \sin\theta + V_2 \cos\theta.
\end{align}

\subsubsection*{Contraction of the Left Region}

We show by mathematical induction that the partial sum on the left side $\sum_{j=1}^{k-1} c_j W_j$ is integrated into $W_k$ by the action of $U_L$.

First, we demonstrate by induction that it is contracted into $W_{k-1}$ by the unitary transformation $U'_{L} = R_{k-2} \dots R_1$.
Here, the generator of the rotation $R_m = \exp(-i \theta_m G_m/2)$ is $G_m = i W_{m+1} W_m$.

First, for the case $m=1$, the initial state is $c_1 W_1 + c_2 W_2 = \epsilon W_1 + \epsilon W_2$.
Using $\theta_1 = \arctan(c_1/c_2) = \arctan(1)$, we have:
\begin{align}
    &R_1 (\epsilon W_1 + \epsilon W_2) R_1^\dagger \\
    &= \epsilon (W_1 \cos\theta_1 + W_2 \sin\theta_1) + \epsilon (-W_1 \sin\theta_1 + W_2 \cos\theta_1) \\
    &= W_1 (\epsilon \cos\theta_1 - \epsilon \sin\theta_1) + W_2 (\epsilon \sin\theta_1 + \epsilon \cos\theta_1) \\
    &= \sqrt{2}\epsilon W_2.
\end{align}
Thus, after the rotation for $m=1$, the coefficients are contracted into $W_2$.

Next, assume that the sum $\sum_{j=1}^{m} c_j W_j$ has been contracted into $\sqrt{m}\epsilon W_m$ by the rotations up to step $m-1$.
Next, we apply $R_m$. The target terms are $\sqrt{m}\epsilon W_m + c_{m+1} W_{m+1}$ (note that $c_{m+1}=\epsilon$).
Using $\theta_m = \arctan(\sqrt{m})$, we obtain:
\begin{align}
    &R_m (\sqrt{m}\epsilon W_m + \epsilon W_{m+1}) R_m^\dagger \\
    &= \sqrt{m}\epsilon (W_m \cos\theta_m + W_{m+1} \sin\theta_m) \\
    &\quad + \epsilon (-W_m \sin\theta_m + W_{m+1} \cos\theta_m)\\
    &=\sqrt{m+1}\epsilon W_{m+1}.
\end{align}

Therefore, by repeating this until $m=k-2$, the sum on the left side is contracted into $W_{k-1}$ as follows:
\begin{align}
    U'_{L} \qty(\sum_{j=1}^{k-1} c_j W_j) (U'_{L})^\dagger = \sqrt{k-1}\epsilon W_{k-1}.
\end{align}

At this point, we have the contribution from the left $\nu_{k-1} W_{k-1}$ (where $\nu_{k-1} = \sqrt{k-1}\epsilon$) and the target term $c_k W_k$.
Applying the final left rotation $R_{k-1}$ with the generator $G_{k-1} = i W_k W_{k-1}$ and $\theta_{k-1}=\arctan(\nu_{k-1}/c_k)$, we obtain:
\begin{align}
    &R_{k-1} (\nu_{k-1} W_{k-1} + c_k W_k) R_{k-1}^\dagger \\
    &= W_{k-1} (\nu_{k-1} \cos\theta_{k-1}- c_k \sin\theta_{k-1})\\
    &\quad + W_k (\nu_{k-1} \sin\theta_{k-1} + c_k \cos\theta_{k-1})\\
    &=\sqrt{(k-1)\epsilon^2 + c_k^2}W_k.
\end{align}
Thus, all terms on the left side are integrated into $W_k$ (if $k=n$, the coefficient becomes $1$ at this step, completing the proof).

\subsubsection*{Contraction of the Right Region}
Similarly, we show that the partial sum on the right side $\sum_{j=k+1}^{n} c_j W_j$ is integrated into $W_k$ by the action of $U_R$.

We apply the rotations $R_m = \exp(-i \theta_m G_m/2)$ while decreasing $m$ from $n-1$ down to $k+1$.
For the contraction from the right side, we define the generator as $G_m = i W_m W_{m+1}$.

First, for the case $m=n-1$, the target is $\epsilon W_{n-1} + \epsilon W_n$. Using $\theta_{n-1} = \arctan(1)$,
\begin{align}
    R_{n-1} (\epsilon W_{n-1} + \epsilon W_n) R_{n-1}^\dagger = \sqrt{2}\epsilon W_{n-1}.
\end{align}
This transfers the component of $W_n$ to $W_{n-1}$.

Now, suppose that the sum from the right end has been contracted into $\sqrt{n-m}\epsilon W_{m+1}$ up to step $m+1$.
To proceed, we apply $R_m$. The target terms are $c_m W_m + \sqrt{n-m}\epsilon W_{m+1}$ (since $m > k$, $c_m=\epsilon$).
Using $\theta_m = \arctan(\sqrt{n-m})$,
\begin{align}
    R_m (\epsilon W_m + \sqrt{n-m}\epsilon W_{m+1}) R_m^\dagger &= W_m \sqrt{\epsilon^2 + (n-m)\epsilon^2} \\
    &= \sqrt{n-m+1}\epsilon W_m.
\end{align}
Repeating this until $m=k+1$, the sum on the right side is contracted into $W_{k+1}$ as follows:
\begin{align}
    U'_{R} \qty(\sum_{j=k+1}^{n} c_j W_j) (U'_{R})^\dagger = \sqrt{n-k}\epsilon W_{k+1}.
\end{align}

Finally, we integrate the result of the right contraction $\nu_{k+1} W_{k+1}$ (where $\nu_{k+1} = \sqrt{n-k}\epsilon$) into the central $W_k$.
We apply the rotation $R_k = \exp(-i \theta_k G_k/2)$ with the generator $G_k = i W_k W_{k+1}$.
The target operators are the term of $W_k$ obtained in Step 1 and the term of $W_{k+1}$ obtained in Step 2.
Using $\theta_k=\arctan\qty( \sqrt{n-k}\epsilon/\sqrt{(k-1)\epsilon^2 + c_k^2} )$,
\begin{align}
    &R_{k} \qty(\sqrt{(k-1)\epsilon^2 + c_k^2}  W_k + \sqrt{n-k}\epsilon  W_{k+1}) R_{k}^\dagger \\
    &= \sqrt{ \qty(\sqrt{(k-1)\epsilon^2 + c_k^2})^2 + \qty(\sqrt{n-k}\epsilon)^2 }W_k\\
    &=W_k.
\end{align}
Therefore, we finally obtain:
\begin{align}
    U O_k U^\dagger = 1 \cdot W_k = \mathcal{E}_k.
\end{align}

\end{proof}

\section{Efficient Encoding Circuit Construction}
\label{app:encode}

In this appendix, we present a concrete construction of the quantum circuit that generates the encoded state $\ket{\psi_y}$ for an odd-parity input $y \in O_n$ in the $(n, n-1)$-QRACs. We further prove that the circuit depth and gate count scale as $O(n^2)$ with respect to $n$.

The encoding procedure consists of two stages: (1) generation of an input-independent reference state $\ket{\psi_{\text{ref}}}$, and (2) a displacement operation using Pauli operators dependent on the input $y$.

\subsection{Algebraic Structure and Translational Symmetry}

Before proceeding to the specific circuit construction, we discuss the geometric symmetry of the encoded states in the $(n, n-1)$-QRACs.
In this subsection, we show that the matrix $A_n$ possesses specific covariance properties under the action of the Pauli group.

First, each encoded state $\ket{\psi_y}$ satisfies an orthogonality relation defined by $A_n$ with respect to its neighborhood set $N(y)$.
We define the bit-flip operator $\mathcal{X}(u)$ on the $n$-qubit Hilbert space, corresponding to the operation of shifting the input $y$ to $y' = y \oplus u$ (where $u \in \{0, 1\}^n$), as follows:
\begin{align}
    \mathcal{X}(u) = \bigotimes_{i=1}^{n} X_i^{u_i}.
\end{align}
Consider the commutation relation between this operator $\mathcal{X}(u)$ and the matrix $A_n = \sum g_l$.
Calculating the conjugation for each term $g_l = Z^{\otimes (l-1)} \otimes X \otimes I_{n-l}$, we obtain:
\begin{align}
    \mathcal{X}(u) g_l \mathcal{X}(u)^\dagger = (-1)^{\alpha_l(u)} g_l,
\end{align}
where
\begin{align}
    \alpha_l(u) \equiv \sum_{j=1}^{l-1} u_j \pmod 2.
\end{align}

To correct this sign disturbance, we similarly consider a phase operator $\mathcal{Z}(v) = \bigotimes_{i=1}^n Z_i^{v_i}$ on the $n$-qubit space. When $\mathcal{Z}(v)$ acts on $g_l$, a new sign flip occurs due to the anti-commutation relation with the $X$ operator at the $l$-th position of $g_l$:
\begin{align}
    \mathcal{Z}(v) g_l \mathcal{Z}(v)^\dagger = (-1)^{v_l} g_l.
\end{align}
Therefore, by defining $\mathcal{D}(u, v) = \mathcal{Z}(v) \mathcal{X}(u)$ and choosing $v$ to satisfy the condition $v_l = \alpha_l(u)$, we can cancel the sign changes for all $l$:
\begin{align}
    \mathcal{D}(u, v) A_n \mathcal{D}(u, v)^\dagger &= \sum_{l=1}^n (-1)^{\alpha_l(u) + v_l} g_l  \\
    &= \sum_{l=1}^n g_l = A_n.
\end{align}
The equation $v_l = \alpha_l(u)$ always has a solution.
Thus, by setting $v_l = \alpha_l(u)$ using an $(n-1)$-bit correction string $v \in \{0, 1\}^{n-1}$, $A_n$ remains invariant over the entire space.

\begin{theorem}\label{th:ref_state}
    The encoded state $\ket{\psi_y}$ for any input $y$ is unitarily equivalent to a fixed reference state $\ket{\psi_{\text{ref}}}$ up to a global phase:
    \begin{align}
        \ket{\psi_y} =(-1)^{v \cdot y}  D(u, v') \ket{\psi_{\text{ref}}},
    \end{align}
    where $D(u, v) =Z(v) X(u)$, with $X(u) = \bigotimes_{i=1}^{n-1} X_i^{u_i}$ and $Z(v) = \bigotimes_{i=1}^{n-1} Z_i^{v_i}$. In addition, $v'_i$ is defined as $v'_i=v_i+v_n$.
\end{theorem}

\begin{proof}
First, for a target input $y$, we define the bit strings $u, v$ constituting the displacement operator $\mathcal{D}(u,v)$ as follows:
\begin{align}
    u &:= y \oplus y_{\text{ref}}, \\
    v_l &:= \alpha_l(u) \quad (\text{for all } l).
\end{align}
From this definition of $u$, $\mathcal{D}(u,v)$ acts to map between $y$ and $y_{\text{ref}}$. Indeed, applying $\mathcal{D}^\dagger(u,v)$ to the computational basis $\ket{y}=\ket{y_1\dots y_n}$ yields:
\begin{align}
    \mathcal{D}^\dagger(u,v) \ket{y} &= \mathcal{X}(u) \mathcal{Z}(v) \ket{y}  \\
    &= (-1)^{v \cdot y} \mathcal{X}(u) \ket{y}  \\
    &= (-1)^{v \cdot y} \ket{y \oplus u}  \\
    &= (-1)^{v \cdot y}\ket{y_{\text{ref}}}.
\end{align}
This shows that $\mathcal{D}^\dagger(u,v)$ maps to the reference basis $\ket{y_{\text{ref}}}$ up to a global phase.

Using this, we have:
\begin{align}
    \ket{\psi_y} &= \frac{1}{\sqrt{n}} \sum_{x \in E_n} \bra{y} A_n \ket{x} \ket{\psi_x}  \\
    &= \frac{1}{\sqrt{n}} \sum_{x \in E_n} \bra{y} \mathcal{D} (u,v)A_n \mathcal{D}^\dagger(u,v) \ket{x} \ket{\psi_x}  \\
    &= (-1)^{v \cdot y} \frac{1}{\sqrt{n}} \sum_{x \in E_n} \bra{y_{\text{ref}}} A_n \mathcal{D}^\dagger(u,v) \ket{x} \ket{\psi_x}.
\end{align}
\begin{widetext}
Inserting the completeness relation $I_{n-1} = \sum_{x' \in E_n} \ket{x'}\bra{x'}+\sum_{y' \in O_n} \ket{y'}\bra{y'}$, we get:
\begin{align}
    \ket{\psi_y} &= (-1)^{v \cdot y} \frac{1}{\sqrt{n}} \sum_{x \in E_n} \sum_{x' \in E_n} \bra{y_{\text{ref}}} A_n \ket{x'} \bra{x'} \mathcal{D}^\dagger (u,v)\ket{x} \ket{\psi_x}  \\
    &= (-1)^{v \cdot y} \frac{1}{\sqrt{n}} \sum_{x' \in E_n} \bra{y_{\text{ref}}} A_n \ket{x'} \left( \sum_{x \in E_n} \bra{x'} \mathcal{D}^\dagger(u,v) \ket{x} \ket{\psi_x} \right) .
\end{align}    
\end{widetext}
Due to the even parity condition $x_n = \bigoplus_{i=1}^{n-1} x_i$, we have:
\begin{align}
    \bra{x'_n} \mathcal{D}_n^\dagger \ket{x_n} &= \bra{x'_n} X_n^{u_n} Z_n^{v_n} \ket{x_n}  \\
    &= (-1)^{v_n x_n} \delta_{x'_n, x_n \oplus u_n}  \\
    &= \left( \prod_{i=1}^{n-1} (-1)^{v_n x_i} \right).
\end{align}
Therefore,
\begin{align}
    &\sum_{x \in E_n} \bra{x'} \mathcal{D}^\dagger(u,v) \ket{x} \ket{\psi_x}\\
    &=\sum_{x \in E_n}\bra{\psi_{x'}} D^\dagger(u,v) \ket{\psi_{x}}\left( \prod_{i=1}^{n-1} (-1)^{v_n x_i} \right)  \ket{\psi_x}\\
    &=\sum_{x \in E_n}\bigotimes_{i=1}^{n-1} Z_i^{v_n}\ket{\psi_x}\bra{\psi_{x}} D(u,v) \ket{\psi_{x'}}\\
    &=\bigotimes_{i=1}^{n-1} Z_i^{v_n} D(u,v) \ket{\psi_{x'}}\\
    &= D(u,v') \ket{\psi_{x'}},
\end{align}
where we set $v'_i=v_i+v_n$. Consequently,
\begin{align}
    \ket{\psi_y} 
    &= (-1)^{v \cdot y} D(u,v') \frac{1}{\sqrt{n}} \sum_{x' \in E_n} (A_n)_{y_{\text{ref}},x'} \ket{\psi_{x'}}\\
    &= (-1)^{v \cdot y} D(u,v')\ket{\psi_{y_{\text{ref}}}}.
\end{align}

\end{proof}

\subsection{Reference State Preparation}

First, consider the encoded state $\ket{\psi_{\text{ref}}}$ corresponding to the reference input $y_{\text{ref}} = (0, 0, \dots, 0, 1) \in O_n$. In the system of $m=n-1$ qubits, this state corresponds to a uniform superposition of the all-zero basis $\ket{0}^{\otimes m}$ and all weight-1 basis states $\ket{e_k}$ ($k=1,\dots,m$), where $\ket{e_k}$ is the state with $\ket{1}$ only at the $k$-th qubit.
\begin{align}
    \ket{\psi_{\text{ref}}} = \frac{1}{\sqrt{n}} \qty( \ket{0}^{\otimes m} + \sum_{k=1}^{m} \ket{e_k} ).
\end{align}
This state can be generated from the initial state $\ket{0}^{\otimes m}$ by sequentially applying $m$ rotation gates as follows.

In each step $k=1, \dots, m$, we apply a controlled rotation gate $C^{k-1}R_y(\theta_k)$ to qubit $q_k$.
This is a multi-controlled rotation gate with the control condition ``all upper bits $q_1, \dots, q_{k-1}$ are $\ket{0}$".
The rotation angle $\theta_k$ is set to satisfy:
\begin{align}
    \theta_k = 2 \arctan \qty( \frac{1}{\sqrt{n - k}} ).
\end{align}
The correctness of this recursive construction can be verified by observing the state transition at each specific step. Let the initial state be $\ket{\Psi_0} = \ket{0}^{\otimes m}$.

\begin{itemize}
    \item $k=1$: \\
    Apply $R_y(\theta_1)$ to the first qubit $q_1$, where $\sin(\theta_1/2) = 1/\sqrt{n}$.
    \begin{align}
        \ket{\Psi_1} &= \qty( \frac{1}{\sqrt{n}}\ket{1}_1 + \sqrt{\frac{n-1}{n}}\ket{0}_1 ) \otimes \ket{0\dots0}_{2\dots m}  \\
        &= \frac{1}{\sqrt{n}}\ket{10\dots0} + \sqrt{\frac{n-1}{n}}\ket{00\dots0}.
    \end{align}
    This fixes the amplitude of $\ket{10\dots0}$ to the desired value $1/\sqrt{n}$.

    \item $k=2$: \\
    Apply $R_y(\theta_2)$ to the second qubit $q_2$ under the control condition ``$q_1=\ket{0}$".
    The first term of $\ket{\Psi_1}$ ($q_1=\ket{1}$) remains unchanged, and only the second term ($q_1=\ket{0}$) undergoes the operation.
    Since $\sin(\theta_2/2) = 1/\sqrt{n-1}$,
    \begin{widetext}
    \begin{align}
        \ket{\Psi_2} &= \frac{1}{\sqrt{n}}\ket{10\dots0} + \sqrt{\frac{n-1}{n}}\ket{0}_1 \otimes \qty( \frac{1}{\sqrt{n-1}}\ket{1}_2 + \sqrt{\frac{n-2}{n-1}}\ket{0}_2 ) \otimes \ket{00\dots 0}_{3\dots m}  \\
        &= \frac{1}{\sqrt{n}}\ket{10\dots0} + \frac{1}{\sqrt{n}}\ket{010\dots0} + \sqrt{\frac{n-2}{n}}\ket{00\dots0}.
    \end{align}
    \end{widetext}
\end{itemize}
Thus, at each step $k$, a new term $1/\sqrt{n}\ket{e_k}$ is generated. By repeating this until the end, we finally obtain the desired uniform superposition state.

\subsection{State Displacement via Pauli Operations}

From Theorem~\ref{th:ref_state}, the encoded state $\ket{\psi_y}$ for any input $y \in O_n$ can be generated by applying Pauli operators $X$ and $Z$ to the reference state $\ket{\psi_{\text{ref}}}$.
\begin{align}
    \ket{\psi_y} &= (-1)^{v \cdot y} D(u,v')\ket{\psi_{\text{ref}}}\\
    &=(-1)^{v \cdot y} \qty( \bigotimes_{j=1}^{n-1} Z_j^{v_j+v_n} ) \qty( \bigotimes_{j=1}^{n-1} X_j^{u_j} ) \ket{\psi_{\text{ref}}},
\end{align}
where the bit string $u \in \{0, 1\}^m$ is determined based on the difference between the input $y$ and the reference input $y_{\text{ref}}$, serving to shift the center position of the superposition in the space. The bit string $v \in \{0, 1\}^m$ acts as a phase factor to correct the sign changes associated with this shift and is efficiently computable classically.

\subsection{Complexity Analysis}

We estimate the computational cost of the above construction. For the generation of the reference state, it is known from previous studies that $C^{k-1}R_y(\theta_k)$ can be implemented with a depth of $O(k)$ even on a 1D chain architecture~\cite{zindorf2025multi}. Therefore, the total CNOT count and depth scale as $O(n^2)$.
On the other hand, for the displacement operation, the application of Pauli operators $X, Z$ involves single-qubit operations performed independently on each qubit, resulting in a circuit depth of $O(1)$.
In conclusion, the overall circuit depth and gate count for encoding scale as $O(n^2)$, demonstrating that exponential resources are not required.

\bibliography{ref.bib}

\begin{thebibliography}{20}%
\makeatletter
\providecommand \@ifxundefined [1]{%
 \@ifx{#1\undefined}
}%
\providecommand \@ifnum [1]{%
 \ifnum #1\expandafter \@firstoftwo
 \else \expandafter \@secondoftwo
 \fi
}%
\providecommand \@ifx [1]{%
 \ifx #1\expandafter \@firstoftwo
 \else \expandafter \@secondoftwo
 \fi
}%
\providecommand \natexlab [1]{#1}%
\providecommand \enquote  [1]{``#1''}%
\providecommand \bibnamefont  [1]{#1}%
\providecommand \bibfnamefont [1]{#1}%
\providecommand \citenamefont [1]{#1}%
\providecommand \href@noop [0]{\@secondoftwo}%
\providecommand \href [0]{\begingroup \@sanitize@url \@href}%
\providecommand \@href[1]{\@@startlink{#1}\@@href}%
\providecommand \@@href[1]{\endgroup#1\@@endlink}%
\providecommand \@sanitize@url [0]{\catcode `\\12\catcode `\$12\catcode `\&12\catcode `\#12\catcode `\^12\catcode `\_12\catcode `\%12\relax}%
\providecommand \@@startlink[1]{}%
\providecommand \@@endlink[0]{}%
\providecommand \url  [0]{\begingroup\@sanitize@url \@url }%
\providecommand \@url [1]{\endgroup\@href {#1}{\urlprefix }}%
\providecommand \urlprefix  [0]{URL }%
\providecommand \Eprint [0]{\href }%
\providecommand \doibase [0]{https://doi.org/}%
\providecommand \selectlanguage [0]{\@gobble}%
\providecommand \bibinfo  [0]{\@secondoftwo}%
\providecommand \bibfield  [0]{\@secondoftwo}%
\providecommand \translation [1]{[#1]}%
\providecommand \BibitemOpen [0]{}%
\providecommand \bibitemStop [0]{}%
\providecommand \bibitemNoStop [0]{.\EOS\space}%
\providecommand \EOS [0]{\spacefactor3000\relax}%
\providecommand \BibitemShut  [1]{\csname bibitem#1\endcsname}%
\let\auto@bib@innerbib\@empty
\bibitem [{\citenamefont {Holevo}(1973)}]{holevo1973}%
  \BibitemOpen
  \bibfield  {author} {\bibinfo {author} {\bibfnamefont {A.~S.}\ \bibnamefont {Holevo}},\ }\bibfield  {title} {\bibinfo {title} {Bounds for the quantity of information transmitted by a quantum communication channel},\ }\href@noop {} {\bibfield  {journal} {\bibinfo  {journal} {Probl. Inf. Transm.}\ }\textbf {\bibinfo {volume} {9}},\ \bibinfo {pages} {177} (\bibinfo {year} {1973})}\BibitemShut {NoStop}%
\bibitem [{\citenamefont {Ambainis}\ \emph {et~al.}(1999)\citenamefont {Ambainis}, \citenamefont {Nayak}, \citenamefont {Ta-Shma},\ and\ \citenamefont {Vazirani}}]{ambainis1999}%
  \BibitemOpen
  \bibfield  {author} {\bibinfo {author} {\bibfnamefont {A.}~\bibnamefont {Ambainis}}, \bibinfo {author} {\bibfnamefont {A.}~\bibnamefont {Nayak}}, \bibinfo {author} {\bibfnamefont {A.}~\bibnamefont {Ta-Shma}},\ and\ \bibinfo {author} {\bibfnamefont {U.}~\bibnamefont {Vazirani}},\ }\bibfield  {title} {\bibinfo {title} {Dense quantum coding and a lower bound for 1-way quantum automata},\ }in\ \href {https://doi.org/10.1145/301250.301347} {\emph {\bibinfo {booktitle} {Proceedings of the 31st Annual ACM Symposium on Theory of Computing}}}\ (\bibinfo  {publisher} {ACM},\ \bibinfo {year} {1999})\ pp.\ \bibinfo {pages} {376--383}\BibitemShut {NoStop}%
\bibitem [{\citenamefont {Ambainis}\ \emph {et~al.}(2002)\citenamefont {Ambainis}, \citenamefont {Nayak}, \citenamefont {Ta-Shma},\ and\ \citenamefont {Vazirani}}]{ambainis2002}%
  \BibitemOpen
  \bibfield  {author} {\bibinfo {author} {\bibfnamefont {A.}~\bibnamefont {Ambainis}}, \bibinfo {author} {\bibfnamefont {A.}~\bibnamefont {Nayak}}, \bibinfo {author} {\bibfnamefont {A.}~\bibnamefont {Ta-Shma}},\ and\ \bibinfo {author} {\bibfnamefont {U.}~\bibnamefont {Vazirani}},\ }\bibfield  {title} {\bibinfo {title} {Dense quantum coding and quantum finite automata},\ }\href {https://doi.org/10.1145/581771.581773} {\bibfield  {journal} {\bibinfo  {journal} {J. ACM}\ }\textbf {\bibinfo {volume} {49}},\ \bibinfo {pages} {496} (\bibinfo {year} {2002})}\BibitemShut {NoStop}%
\bibitem [{\citenamefont {Klauck}(2000)}]{klauck2000}%
  \BibitemOpen
  \bibfield  {author} {\bibinfo {author} {\bibfnamefont {H.}~\bibnamefont {Klauck}},\ }\bibfield  {title} {\bibinfo {title} {On quantum and probabilistic communication: {Las} {Vegas} and one-way protocols},\ }in\ \href {https://doi.org/10.1145/335305.335396} {\emph {\bibinfo {booktitle} {Proceedings of the 32nd Annual ACM Symposium on Theory of Computing}}}\ (\bibinfo  {publisher} {ACM},\ \bibinfo {year} {2000})\ pp.\ \bibinfo {pages} {644--651}\BibitemShut {NoStop}%
\bibitem [{\citenamefont {Aaronson}(2004)}]{aaronson2004}%
  \BibitemOpen
  \bibfield  {author} {\bibinfo {author} {\bibfnamefont {S.}~\bibnamefont {Aaronson}},\ }\bibfield  {title} {\bibinfo {title} {Limitations of quantum advice and one-way communication},\ }in\ \href {https://doi.org/10.1109/CCC.2004.1313854} {\emph {\bibinfo {booktitle} {Proceedings of the 19th IEEE Annual Conference on Computational Complexity}}}\ (\bibinfo  {publisher} {IEEE},\ \bibinfo {year} {2004})\ pp.\ \bibinfo {pages} {320--332}\BibitemShut {NoStop}%
\bibitem [{\citenamefont {Kerenidis}\ and\ \citenamefont {de~Wolf}(2003)}]{kerenidis2003}%
  \BibitemOpen
  \bibfield  {author} {\bibinfo {author} {\bibfnamefont {I.}~\bibnamefont {Kerenidis}}\ and\ \bibinfo {author} {\bibfnamefont {R.}~\bibnamefont {de~Wolf}},\ }\bibfield  {title} {\bibinfo {title} {Exponential lower bound for 2-query locally decodable codes via a quantum argument},\ }in\ \href {https://doi.org/10.1145/780542.780560} {\emph {\bibinfo {booktitle} {Proceedings of the 35th Annual ACM Symposium on Theory of Computing (STOC '03)}}}\ (\bibinfo  {publisher} {ACM},\ \bibinfo {address} {San Diego, CA, USA},\ \bibinfo {year} {2003})\ pp.\ \bibinfo {pages} {106--115}\BibitemShut {NoStop}%
\bibitem [{\citenamefont {Paw{\l}owski}\ and\ \citenamefont {Brunner}(2011)}]{pawlowski2011}%
  \BibitemOpen
  \bibfield  {author} {\bibinfo {author} {\bibfnamefont {M.}~\bibnamefont {Paw{\l}owski}}\ and\ \bibinfo {author} {\bibfnamefont {N.}~\bibnamefont {Brunner}},\ }\bibfield  {title} {\bibinfo {title} {Semi-device-independent security of one-way quantum key distribution},\ }\href {https://doi.org/10.1103/PhysRevA.84.010302} {\bibfield  {journal} {\bibinfo  {journal} {Phys. Rev. A}\ }\textbf {\bibinfo {volume} {84}},\ \bibinfo {pages} {010302} (\bibinfo {year} {2011})}\BibitemShut {NoStop}%
\bibitem [{\citenamefont {Li}\ \emph {et~al.}(2011)\citenamefont {Li}, \citenamefont {Yin}, \citenamefont {Wu}, \citenamefont {Zou}, \citenamefont {Wang}, \citenamefont {Chen}, \citenamefont {Guo},\ and\ \citenamefont {Han}}]{li2011}%
  \BibitemOpen
  \bibfield  {author} {\bibinfo {author} {\bibfnamefont {H.-W.}\ \bibnamefont {Li}}, \bibinfo {author} {\bibfnamefont {Z.-Q.}\ \bibnamefont {Yin}}, \bibinfo {author} {\bibfnamefont {Y.-C.}\ \bibnamefont {Wu}}, \bibinfo {author} {\bibfnamefont {X.-B.}\ \bibnamefont {Zou}}, \bibinfo {author} {\bibfnamefont {S.}~\bibnamefont {Wang}}, \bibinfo {author} {\bibfnamefont {W.}~\bibnamefont {Chen}}, \bibinfo {author} {\bibfnamefont {G.-C.}\ \bibnamefont {Guo}},\ and\ \bibinfo {author} {\bibfnamefont {Z.-F.}\ \bibnamefont {Han}},\ }\bibfield  {title} {\bibinfo {title} {Semi-device-independent random-number expansion without entanglement},\ }\href {https://doi.org/10.1103/PhysRevA.84.034301} {\bibfield  {journal} {\bibinfo  {journal} {Phys. Rev. A}\ }\textbf {\bibinfo {volume} {84}},\ \bibinfo {pages} {034301} (\bibinfo {year} {2011})}\BibitemShut {NoStop}%
\bibitem [{\citenamefont {Brunner}\ \emph {et~al.}(2013)\citenamefont {Brunner}, \citenamefont {Navascu{\'e}s},\ and\ \citenamefont {V{\'e}rtesi}}]{brunner2013}%
  \BibitemOpen
  \bibfield  {author} {\bibinfo {author} {\bibfnamefont {N.}~\bibnamefont {Brunner}}, \bibinfo {author} {\bibfnamefont {M.}~\bibnamefont {Navascu{\'e}s}},\ and\ \bibinfo {author} {\bibfnamefont {T.}~\bibnamefont {V{\'e}rtesi}},\ }\bibfield  {title} {\bibinfo {title} {Dimension witnesses and quantum state discrimination},\ }\href {https://doi.org/10.1103/PhysRevLett.110.150501} {\bibfield  {journal} {\bibinfo  {journal} {Phys. Rev. Lett.}\ }\textbf {\bibinfo {volume} {110}},\ \bibinfo {pages} {150501} (\bibinfo {year} {2013})}\BibitemShut {NoStop}%
\bibitem [{\citenamefont {Tavakoli}\ \emph {et~al.}(2021)\citenamefont {Tavakoli}, \citenamefont {Pauwels}, \citenamefont {Woodhead},\ and\ \citenamefont {Pironio}}]{tavakoli2021}%
  \BibitemOpen
  \bibfield  {author} {\bibinfo {author} {\bibfnamefont {A.}~\bibnamefont {Tavakoli}}, \bibinfo {author} {\bibfnamefont {J.}~\bibnamefont {Pauwels}}, \bibinfo {author} {\bibfnamefont {E.}~\bibnamefont {Woodhead}},\ and\ \bibinfo {author} {\bibfnamefont {S.}~\bibnamefont {Pironio}},\ }\bibfield  {title} {\bibinfo {title} {Correlations in entanglement-assisted prepare-and-measure scenarios},\ }\href {https://doi.org/10.1103/PRXQuantum.2.040357} {\bibfield  {journal} {\bibinfo  {journal} {PRX Quantum}\ }\textbf {\bibinfo {volume} {2}},\ \bibinfo {pages} {040357} (\bibinfo {year} {2021})}\BibitemShut {NoStop}%
\bibitem [{\citenamefont {Farkas}\ and\ \citenamefont {Kaniewski}(2019)}]{farkas2019}%
  \BibitemOpen
  \bibfield  {author} {\bibinfo {author} {\bibfnamefont {M.}~\bibnamefont {Farkas}}\ and\ \bibinfo {author} {\bibfnamefont {J.}~\bibnamefont {Kaniewski}},\ }\bibfield  {title} {\bibinfo {title} {Self-testing mutually unbiased bases in the prepare-and-measure scenario},\ }\href {https://doi.org/10.1103/PhysRevA.99.032316} {\bibfield  {journal} {\bibinfo  {journal} {Phys. Rev. A}\ }\textbf {\bibinfo {volume} {99}},\ \bibinfo {pages} {032316} (\bibinfo {year} {2019})}\BibitemShut {NoStop}%
\bibitem [{\citenamefont {Miao}\ \emph {et~al.}(2022)\citenamefont {Miao}, \citenamefont {Liu}, \citenamefont {Sun}, \citenamefont {Ning}, \citenamefont {Li},\ and\ \citenamefont {Guo}}]{miao2022}%
  \BibitemOpen
  \bibfield  {author} {\bibinfo {author} {\bibfnamefont {R.-H.}\ \bibnamefont {Miao}}, \bibinfo {author} {\bibfnamefont {Z.-D.}\ \bibnamefont {Liu}}, \bibinfo {author} {\bibfnamefont {Y.-N.}\ \bibnamefont {Sun}}, \bibinfo {author} {\bibfnamefont {C.-X.}\ \bibnamefont {Ning}}, \bibinfo {author} {\bibfnamefont {C.-F.}\ \bibnamefont {Li}},\ and\ \bibinfo {author} {\bibfnamefont {G.-C.}\ \bibnamefont {Guo}},\ }\bibfield  {title} {\bibinfo {title} {High-dimensional multi-input quantum random access codes and mutually unbiased bases},\ }\href {https://doi.org/10.1103/PhysRevA.106.042418} {\bibfield  {journal} {\bibinfo  {journal} {Phys. Rev. A}\ }\textbf {\bibinfo {volume} {106}},\ \bibinfo {pages} {042418} (\bibinfo {year} {2022})}\BibitemShut {NoStop}%
\bibitem [{\citenamefont {Nayak}(1999)}]{nayak1999}%
  \BibitemOpen
  \bibfield  {author} {\bibinfo {author} {\bibfnamefont {A.}~\bibnamefont {Nayak}},\ }\bibfield  {title} {\bibinfo {title} {Optimal lower bounds for quantum automata and random access codes},\ }in\ \href {https://doi.org/10.1109/SFFCS.1999.814608} {\emph {\bibinfo {booktitle} {Proceedings of the 40th Annual Symposium on Foundations of Computer Science}}}\ (\bibinfo  {publisher} {IEEE},\ \bibinfo {year} {1999})\ pp.\ \bibinfo {pages} {369--376}\BibitemShut {NoStop}%
\bibitem [{\citenamefont {Farkas}\ \emph {et~al.}(2025)\citenamefont {Farkas}, \citenamefont {Miklin},\ and\ \citenamefont {Tavakoli}}]{farkas2025}%
  \BibitemOpen
  \bibfield  {author} {\bibinfo {author} {\bibfnamefont {M.}~\bibnamefont {Farkas}}, \bibinfo {author} {\bibfnamefont {N.}~\bibnamefont {Miklin}},\ and\ \bibinfo {author} {\bibfnamefont {A.}~\bibnamefont {Tavakoli}},\ }\bibfield  {title} {\bibinfo {title} {Simple and general bounds on quantum random access codes},\ }\href {https://doi.org/10.22331/q-2025-02-25-1643} {\bibfield  {journal} {\bibinfo  {journal} {Quantum}\ }\textbf {\bibinfo {volume} {9}},\ \bibinfo {pages} {1643} (\bibinfo {year} {2025})}\BibitemShut {NoStop}%
\bibitem [{\citenamefont {Tavakoli}\ \emph {et~al.}(2024)\citenamefont {Tavakoli}, \citenamefont {Pozas-Kerstjens}, \citenamefont {Brown},\ and\ \citenamefont {Ara{\'u}jo}}]{tavakoli2024}%
  \BibitemOpen
  \bibfield  {author} {\bibinfo {author} {\bibfnamefont {A.}~\bibnamefont {Tavakoli}}, \bibinfo {author} {\bibfnamefont {A.}~\bibnamefont {Pozas-Kerstjens}}, \bibinfo {author} {\bibfnamefont {P.}~\bibnamefont {Brown}},\ and\ \bibinfo {author} {\bibfnamefont {M.}~\bibnamefont {Ara{\'u}jo}},\ }\bibfield  {title} {\bibinfo {title} {Semidefinite programming relaxations for quantum correlations},\ }\href {https://doi.org/10.1103/RevModPhys.96.045006} {\bibfield  {journal} {\bibinfo  {journal} {Rev. Mod. Phys.}\ }\textbf {\bibinfo {volume} {96}},\ \bibinfo {pages} {045006} (\bibinfo {year} {2024})}\BibitemShut {NoStop}%
\bibitem [{\citenamefont {Navascu{\'e}s}\ \emph {et~al.}(2015)\citenamefont {Navascu{\'e}s}, \citenamefont {Feix}, \citenamefont {Ara{\'u}jo},\ and\ \citenamefont {V{\'e}rtesi}}]{navascues2015}%
  \BibitemOpen
  \bibfield  {author} {\bibinfo {author} {\bibfnamefont {M.}~\bibnamefont {Navascu{\'e}s}}, \bibinfo {author} {\bibfnamefont {A.}~\bibnamefont {Feix}}, \bibinfo {author} {\bibfnamefont {M.}~\bibnamefont {Ara{\'u}jo}},\ and\ \bibinfo {author} {\bibfnamefont {T.}~\bibnamefont {V{\'e}rtesi}},\ }\bibfield  {title} {\bibinfo {title} {Characterizing finite-dimensional quantum behavior},\ }\href {https://doi.org/10.1103/PhysRevA.92.042117} {\bibfield  {journal} {\bibinfo  {journal} {Phys. Rev. A}\ }\textbf {\bibinfo {volume} {92}},\ \bibinfo {pages} {042117} (\bibinfo {year} {2015})}\BibitemShut {NoStop}%
\bibitem [{\citenamefont {Imamichi}\ and\ \citenamefont {Raymond}(2018)}]{imamichi2018}%
  \BibitemOpen
  \bibfield  {author} {\bibinfo {author} {\bibfnamefont {T.}~\bibnamefont {Imamichi}}\ and\ \bibinfo {author} {\bibfnamefont {R.}~\bibnamefont {Raymond}},\ }\href@noop {} {\bibinfo {title} {Constructions of quantum random access codes}},\ \bibinfo {howpublished} {Presented at the 18th Asian Quantum Information Science Conference (AQIS'18)} (\bibinfo {year} {2018}),\ \bibinfo {note} {available at \url{http://www.ngc.is.ritsumei.ac.jp/~ger/static/AQIS18/OnlineBooklet/122}}\BibitemShut {NoStop}%
\bibitem [{\citenamefont {Man{\v{c}}inska}\ and\ \citenamefont {Storgaard}(2022)}]{mancinska2022}%
  \BibitemOpen
  \bibfield  {author} {\bibinfo {author} {\bibfnamefont {L.}~\bibnamefont {Man{\v{c}}inska}}\ and\ \bibinfo {author} {\bibfnamefont {S.~A.~L.}\ \bibnamefont {Storgaard}},\ }\bibfield  {title} {\bibinfo {title} {The geometry of {Bloch} space in the context of quantum random access codes},\ }\href {https://doi.org/10.1007/s11128-022-03470-4} {\bibfield  {journal} {\bibinfo  {journal} {Quantum Inf. Process.}\ }\textbf {\bibinfo {volume} {21}},\ \bibinfo {pages} {132} (\bibinfo {year} {2022})}\BibitemShut {NoStop}%
\bibitem [{\citenamefont {Kondo}\ \emph {et~al.}(2025)\citenamefont {Kondo}, \citenamefont {Sato}, \citenamefont {Raymond},\ and\ \citenamefont {Yamamoto}}]{kondo2025}%
  \BibitemOpen
  \bibfield  {author} {\bibinfo {author} {\bibfnamefont {R.}~\bibnamefont {Kondo}}, \bibinfo {author} {\bibfnamefont {Y.}~\bibnamefont {Sato}}, \bibinfo {author} {\bibfnamefont {R.}~\bibnamefont {Raymond}},\ and\ \bibinfo {author} {\bibfnamefont {N.}~\bibnamefont {Yamamoto}},\ }\bibfield  {title} {\bibinfo {title} {Recursive quantum relaxation for combinatorial optimization problems},\ }\href {https://doi.org/10.22331/q-2025-01-15-1594} {\bibfield  {journal} {\bibinfo  {journal} {Quantum}\ }\textbf {\bibinfo {volume} {9}},\ \bibinfo {pages} {1594} (\bibinfo {year} {2025})}\BibitemShut {NoStop}%
\bibitem [{\citenamefont {Zindorf}\ and\ \citenamefont {Bose}(2025)}]{zindorf2025multi}%
  \BibitemOpen
  \bibfield  {author} {\bibinfo {author} {\bibfnamefont {B.}~\bibnamefont {Zindorf}}\ and\ \bibinfo {author} {\bibfnamefont {S.}~\bibnamefont {Bose}},\ }\href@noop {} {\bibinfo {title} {Multi-controlled quantum gates in linear nearest neighbor}} (\bibinfo {year} {2025}),\ \Eprint {https://arxiv.org/abs/2506.00695} {arXiv:2506.00695 [quant-ph]} \BibitemShut {NoStop}%
\end{thebibliography}%

\end{document}